%% file: main.tex
\documentclass[11pt,letterpaper]{article}

\input{packages.tex}
\input{environments.tex}

\input{macros.tex}

\newif\ifanonymous
\anonymousfalse

\newcommand{\papertitle}{Distributed Lower Bounds via Automatic Self-Reduction}

\begin{document}

% Cover
\begin{flushleft}
    \huge\bfseries
    \papertitle
\end{flushleft}
\smallskip

\newcommand{\myaff}[1]{\,$\cdot$\, {\small #1}\par\medskip}

\newenvironment{myabstract}
{\list{}{\listparindent 1.5em
        \itemindent    \listparindent
        \leftmargin    0cm
        \rightmargin   0cm
        \parsep        0pt}%
    \item\relax}
{\endlist}

\newenvironment{mycover}
{\list{}{\listparindent 0pt
        \itemindent    \listparindent
        \leftmargin    0cm
        \rightmargin   1.5cm
        \parsep        0pt}%
    \raggedright
    \item\relax}
{\endlist}

\begin{mycover}
\ifanonymous
    \textbf{Anonymous Authors}\par
\else
    \textbf{Alkida Balliu}
    \myaff{Gran Sasso Science Institute}

     \textbf{Francesco d'Amore}
    \myaff{Gran Sasso Science Institute}

     \textbf{Dennis Olivetti}
    \myaff{Gran Sasso Science Institute}
\fi
\bigskip
\end{mycover}

\begin{myabstract}
\noindent\textbf{Abstract.}
The development of round elimination into a general-purpose technique [PODC 2019] marked a turning point in our understanding of the hardness of many graph problems in the distributed setting and led to several breakthrough results.
However, the round elimination technique seems unable to yield randomized lower bounds of $\omega(\log \log n)$ rounds as a function of the number $n$ of nodes.

Very recently, Khoury and Schild [FOCS 2025] introduced a new technique called round elimination via \emph{self-reduction}, which bypasses the limitations of classical round elimination. Using this approach, the authors show that any randomized algorithm for maximal matching requires $\Omega(\sqrt{\log n})$ rounds in the LOCAL model.
Their elegant technique is, in some respects, similar to classical round elimination while being fundamentally different in others. However, it is tailored specifically to maximal matching rather than being applicable to a broad class of problems.

In this paper, we show that \emph{self-reduction} is, in fact, a \emph{special case} of classical round elimination, thereby turning it into a general-purpose approach.
In particular, we introduce a new way to measure the error of an algorithm and show that, under this new measure, classical round elimination can indeed yield $\omega(\log \log n)$ randomized lower bounds. More specifically, we identify a large class of problems for which this improvement is entirely black-box: once a problem is shown to belong to the class, stronger randomized lower bounds follow automatically from the classical round-elimination framework.
As an application, we prove $\Omega(\sqrt{\log n})$ randomized lower bounds for a range of graph problems, namely, maximal matching on regular \emph{$2$-colored} graphs, $\frac{1}{k}$-integral matching, and maximal $H$-packing.

\end{myabstract}

\thispagestyle{empty}
\setcounter{page}{0}
\newpage

\input{trunk/notes.tex}

\section*{AI Methodology}

All the main ideas in the paper (the new measure of error, the new definition of the $(T-1)$-round algorithm, the considered class of problems, and the applications) were discovered by humans.

A preliminary version of the proof of \Cref{lem:onestep-singlenode-simplified} was discovered by humans, but LLMs (ChatGPT 5.6) were able to provide a simpler proof, which we then further simplified and used.
LLMs then helped to extend the idea of \Cref{lem:onestep-singlenode-simplified} to the $b$-dominant case, although a similar idea had already appeared in \cite[Lemma 3.8]{balliu-casagrande-etal-2026-new-hardness-results-for}.

The paper was entirely written by humans; LLMs were used to fix typos, grammatical errors, small mathematical flaws, and to create figures.

\ifanonymous
\else
\section*{Acknowledgments}
This work was partially supported by MUR, Fondo Italiano per la Scienza (FIS3), Award number: D53C25002470001: FIS-2024-03606 - A Robust Theory of Distributed Computation - CUP: D53C25002470001.
Francesco d'Amore was supported by the project Decreto MUR n. 47/2025,
CUP: D13C25000750001.
\fi

\printbibliography

\appendix

\section{Existence of high-girth graphs with the desired properties}

In this section, we prove the following lemmas.

\begin{lemma}\label{lem:app:high-girth-independent}
    There exist positive constants \(\alpha, \Delta_0, \varepsilon\) for which the following statement holds.
    For every sufficiently large \(n\), and every \(\Delta\) satisfying that \(\Delta_0 \le \Delta \le n^{1/4}\) and that \(n\Delta\) is even, there exists a \(\Delta\)-regular graph \(G\) of \(n\) nodes that satisfies the following properties:
    \begin{enumerate}
        \item The girth of \(G\) is at least \(\varepsilon \log_\Delta n\).
        \item The independence number of \(G\) is at most \(\alpha (n/\Delta) \ln \Delta\).
    \end{enumerate}
\end{lemma}
\begin{proof}
    Let \(\gamma > 0\) be a sufficiently-small auxiliary constant, and define \(g = \floor{\gamma \log_\Delta n}\).
    We employ the probabilistic method on random regular simple graphs, and prove that with positive probability a graph has both claimed properties, for \(\varepsilon = \gamma/2\) and \(\alpha = 2C\) for a suitable constant \(C > 0\).

    We sample uniformly at random a graph \(G'\) of \(n\) nodes that is simple and \(\Delta\)-regular.
    By \cite[main theorem]{frieze-luczak-1992}, there exist positive constants \(C,\Delta_0,N\) such that, for every \(n \ge N\) and every \(\Delta_0 \le \Delta \le n^{1/4}\) with \(n\Delta\) even, it holds that, with probability \(1 - o(1)\), the independence number of \(G'\) is at most \(C (n/\Delta) \ln \Delta\).
    By increasing \(\Delta_0\) if necessary, assume that \(\Delta_0\ge3\).

    We have two cases.
    First, suppose that \(\gamma \log_\Delta n < 4\).
    Since the bound on the independence number holds with positive probability, there exists a realization \(G''\) of \(G'\) that satisfies the bound.
    Since the graph is simple, \(G''\) has girth at least \(3\), while \((\gamma/2)\log_\Delta n < 2\).
    Thus, \(G''\) already has girth at least \((\gamma/2)\log_\Delta n\), and the claim follows.
    In the second case, \(\gamma \log_\Delta n \ge 4\).
    In particular, \(g \ge \max\{3, (\gamma/2)\log_\Delta n\}\).

    We now count the number of small cycles in \(G'\).
    Let \(X_r\) denote the number of cycles of length exactly \(r\), and define \(X_{<g} = \sum_{r = 3}^{g-1}X_r\).
    Observe that \(X_{<g}\) counts the number of cycles of length strictly less than \(g\).

    Let us define the quantity
    \[
        R_r = \floor*{\max\left\{\frac{(\Delta-1)^r}{r}, \log n\right\}}.
    \]
    By \cite[Lemma~1]{mckay-wormald-wysocka-2004}, applied with the set of cycle lengths \(\{3,\ldots,g-1\}\), under the hypothesis \((\Delta-1)^{2g-1} = o(n)\) stated in \cite[Eq.~(1.1)]{mckay-wormald-wysocka-2004}, it holds with probability \(1-o(1)\) that \(X_r \le R_r\) for every \(3 \le r < g\).
    In our case, we have
    \[
        (\Delta - 1)^{2g - 1} \le \Delta^{2g} \le n^{2\gamma} = o(n).
    \]
    Moreover,
    \[
        \sum_{r = 3}^{g-1} \frac{(\Delta-1)^r}{r} \le \sum_{r = 3}^{g-1}(\Delta-1)^r \le (\Delta-1)^g \le n^\gamma,
    \]
    where in the second inequality we used \(\Delta \ge 3\).
    Consequently, with probability \(1-o(1)\), for sufficiently large \(n\), it holds that
    \[
        X_{<g} \le \sum_{r = 3}^{g-1}R_r = O\left(g \log n + \sum_{r = 3}^{g-1}\frac{(\Delta-1)^r}{r}\right) = O\left(g\log n + \Delta^g\right) = O\left(g \log n + n^\gamma\right) \le n^{4\gamma}.
    \]
    Since the bound on the number of cycles and the bound on the independence number both hold with probability \(1 - o(1)\), there is positive probability that \(G'\) realizes both and, hence, a graph \(G''\) satisfying both bounds exists.

    We now want to remove these short cycles by keeping the graph \(\Delta\)-regular.

    \paragraph{The final graph \(G\).}
    To construct the final graph \(G\), we follow the degree-preserving edge-switching procedure used by Alon~\cite[Lemma~2.1]{alon-2010-on-constant-time-approximation-of-parameters}.
    This procedure is iterative.
    The initial graph is \(F = G''\).
    As long as \(F\) contains a cycle \(K\) of length strictly less than \(g\), we pick an edge \(xy\) in \(K\).
    We then choose another edge \(uv\) in \(F\) such that
    \begin{enumerate}
        \item[(a)] The distance (in \(F\)) between \(\{u,v\}\) and \(K\) is at least \(g\).
        \item[(b)] The edge \(uv\) does not belong to any cycle of \(F\) of length strictly less than \(g\).
    \end{enumerate}
    We define the new graph \(F'\) by deleting edges \(xy\) and \(uv\), while we add edges \(xv\) and \(yu\).
    This operation preserves \(\Delta\)-regularity and simplicity. Indeed, property (a) implies that \(u,v\notin V(K)\) and that neither \(xv\) nor \(yu\) is already an edge of \(F\).
    We now show that such an edge always exists and that this operation strictly decreases the number of cycles of length less than \(g\).
    Then we set \(F = F'\) and we repeat the procedure until no short cycles exist.

    \paragraph{Existence of \(uv\).}
    We prove this fact inductively.
    The current graph contains at most \(n^{4\gamma}\) cycles of length less than \(g\) (the base case is proved by construction of \(G''\)).
    Since \(K\) has at most \(g-1\) nodes and the maximum degree is \(\Delta\), the number of nodes at distance less than \(g\) from \(K\) is at most
    \[
        g(1 + \Delta + \ldots + \Delta^{g-1}) \le g \Delta^g.
    \]
    By \(\Delta\)-regularity, the number of edges incident to these nodes is at most \(g \Delta^{g+1}\).
    Furthermore, the number of edges that belong to cycles of length less than \(g\) is at most \(gn^{4\gamma}\).
    Therefore, the number of edges that fail at least one of properties (a) and (b) is at most \(gn^{4\gamma} + g\Delta^{g+1}\).
    However, the graph has \(n\Delta/2\) edges, and
    \[
        \frac{n\Delta}{2} - gn^{4\gamma} - g\Delta^{g+1} \ge n^{1-5\gamma} \Delta
    \]
    for sufficiently small \(\gamma\) and sufficiently large \(n\).
    In particular, at every iteration of the edge-switching procedure there exists an edge \(uv\) satisfying both properties (a) and (b).

    \paragraph{Switching edges strictly decreases the number of short cycles.}
    Suppose \(F'\) contains a new cycle \(D\) of length less than \(g\).
    Then, \(D\) must contain at least one of the new edges \(xv\) and \(yu\).
    Suppose it contains \(xv\) but not \(yu\).
    Removing \(xv\) from \(D\) gives an \(x\)-to-\(v\) path in \(F\) of length at most \(g-2\).
    Since \(x \in K\), this contradicts property (a).
    The same argument holds if \(D\) contains \(yu\) but not \(xv\).
    Suppose now that \(D\) contains both \(xv\) and \(yu\).
    Removing these two edges splits \(D\) into two paths, all of whose edges belong to \(F\).
    If these paths connect \(x\) to \(y\) and \(u\) to \(v\), then property (b)
    implies that the \(u\)-to-\(v\) path has length at least \(g-1\), contradicting \(\abs{D} < g\).
    If, instead, the paths connect \(x\) to \(u\) and \(y\) to \(v\), property (a) implies that both paths have length at least \(g\), again contradicting \(\abs{D} < g\).

    Thus, the edge-switching creates no new short cycle and deletes \(K\), decreasing the number of short cycles.
    The procedure terminates after at most \(n^{4\gamma}\) iterations; let \(G\) be the resulting graph.
    By construction, \(G\) contains no cycle of length strictly less than \(g\), and hence
    \[
        \operatorname{girth}(G)\ge g
        \ge \frac{\gamma}{2}\log_\Delta n.
    \]
    Thus, \(G\) satisfies the first property of the lemma with \(\varepsilon = \gamma/2\).

    \paragraph{Small independence number.}
    Let \(I\) be an arbitrary independent set of \(G\).
    Every edge of \(G''\) with both endpoints in \(I\) must be one of the at most \(2n^{4\gamma}\) edges removed during the edge-switching procedure.
    By deleting from \(I\) one endpoint of every such edge, we obtain an independent set \(I'\) of \(G''\) satisfying \(\abs{I'} \ge \abs{I} - 2n^{4\gamma}\).
    Therefore,
    \[
        \abs{I} \le C\frac{n \ln \Delta}{\Delta} + 2n^{4\gamma} \le 2C \frac{n \ln \Delta}{\Delta},
    \]
    where the last inequality holds by choosing \(\gamma\) sufficiently small and using \(\Delta \le n^{1/4}\).
    Thus, the claim follows by setting \(\alpha = 2C\).
\end{proof}

\begin{lemma}\label{lem:app:bipartite-high-girth}
    There exist a sufficiently small constant \(c > 0\) and a sufficiently large constant \(C > 0\) such that, for all sufficiently large even \(n > 0\), and for all \(3 \le \Delta \le n^{1/4}\), there exists a bipartite \(\Delta\)-regular graph \(G = (V, E)\) on \(n\) nodes with the following properties:
    \begin{enumerate}
        \item The girth of \(G\) is at least \( c \log_\Delta n\).
        \item Every maximal matching leaves at most \(C (n / \Delta) \ln \Delta\) unmatched nodes.
    \end{enumerate}
\end{lemma}
\begin{proof}
    Let \(L\) and \(R\) be two disjoint sets of \(n/2\) nodes each.
    Fix an auxiliary constant \(0<\gamma<1/12\), and let
    \[
        g=\max\left\{3,\floor{\gamma\log_\Delta n}\right\}.
    \]
    We first construct a random bipartite \(\Delta\)-regular multigraph \(G'\) on the node set \(L \cup R\) such that, with positive probability, \(G'\) satisfies the following two properties:
    \begin{enumerate}
        \item If \(A \subseteq L\) and \(B \subseteq R\) are such that \(\abs{A} = \abs{B} \ge C (n/(2\Delta)) \ln \Delta\), then there exists at least one edge connecting \(A\) to \(B\).
        \item The number of cycles of length strictly less than \(g\) is at most \(n^{2/3}\).
    \end{enumerate}
    The multigraph \(G'\) is constructed as follows: 
    Attach \(\Delta\) half-edges to every node of \(L \cup R\).
    Let \(H_L\) (\(H_R\)) be the set of half-edges attached to \(L\) (\(R\)).
    It holds that \(\abs{H_L} = \abs{H_R} = n\Delta / 2\).
    Then, pick a perfect matching between \(H_L\) and \(H_R\) uniformly at random.
    Each pair of matched half-edges produces an edge in the multigraph \(G'\), which is trivially \(\Delta\)-regular.

    We now prove property 1.
    Let \(s=\ceil{C(n/(2\Delta))\ln\Delta}\).
    If \(s>n/2\), property 1 is vacuous, so assume that \(s\le n/2\).
    Fix \(A \subseteq L\) and \(B \subseteq R\) such that \(\abs{A}=\abs{B}=s\).
    Let \(H_L[A]\) (\(H_R[B]\)) be the subset of \(H_L\) (\(H_R\)) of half-edges of \(A\) (\(B\)).
    Note that \(\abs{H_L[A]} = \abs{H_R[B]} = s \Delta\).
    Now, order the elements of \(H_L[A]\) arbitrarily, and reveal them one by one, together with their matched half-edges in \(H_R\).
    The random number of elements of \(H_L[A]\) that are matched with an element of \(H_R[B]\) is described by a hypergeometric distribution \(\textrm{Hypergeometric}(n \Delta / 2, s \Delta, s\Delta)\).
    The probability that no element of \(H_L[A]\) is matched with an element of \(H_R[B]\) is therefore 
    \[
        \prod_{i = 0}^{s\Delta - 1} \frac{\frac{n\Delta}{2} - s\Delta - i}{\frac{n\Delta}{2} - i} \le \left(1 - \frac{s}{\frac{n}{2}}\right)^{s\Delta} \le \exp\left[- 2 s^2 \frac{\Delta}{n} \right],
    \]
    where the latter inequality holds by the known inequality $1-x\le e^{-x}$ for all $0\le x\le 1$.
    There are \(\binom{n/2}{s}\) choices of \(A\), and \(\binom{n/2}{s}\) choices of \(B\).
    By the union bound, the probability that there exist subsets \(A \subseteq L, B \subseteq R\) of size \(s\) such that no element of \(H_L[A]\) is matched to an element of \(H_R[B]\) is at most
    \begin{align*}
        \binom{\frac{n}{2}}{s}^2 \exp\left[- 2 s^2 \frac{\Delta}{n} \right] & \le \left(\frac{e n}{2s}\right)^{2s} \exp\left[- 2 s^2 \frac{\Delta}{n} \right] \\
        & = \exp\left[-s\left(2s \frac{\Delta}{n} - 2\ln \frac{en}{2s}\right)\right].
    \end{align*}
    By the definition of \(s\), and since \(C\) is sufficiently large, the latter probability is at most
    \begin{align*}
        \exp\left[-s\left(C \ln \Delta - 2\ln\Delta - 2\right)\right] & \le \exp\left[- n^{3/4}\right],
    \end{align*}
    where we exploited the fact that \(C\) is sufficiently large, and that \(3 \le \Delta \le n^{1/4}\).
    Moreover, if one could find \(A\subseteq L,B\subseteq R\) of size strictly larger than \(s\) with no edge between \(A\) and \(B\), then any two subsets \(A' \subseteq A\), \(B' \subseteq B\) of size exactly \(s\) are not adjacent.
    Therefore, with probability at least \(1 - \exp[- n^{3/4}]\), property 1 holds.

    Now we prove property 2.
    Let \(X_{<g}\) count the number of cycles of length strictly less than \(g\).
    If we define \(X_r\) to be the number of cycles of length exactly \(r\), then \(X_{<g}={\sum_{r = 2}^{g-1} X_r}\).
    Here, a cycle of length \(2\) consists of a pair of parallel edges.
    Since \(G'\) is bipartite, \(X_r = 0\) for all odd values of \(r\).
    To bound \(\expect{X_r}\), we consider rooted oriented representations of cycles of length \(r\), for any \(r = 2k\).
    A cycle can be represented as a sequence \(u_1 v_1 u_2 v_2 u_3 \ldots u_k v_k u_{k+ 1}\), where \(u_{k+1} = u_1\), \(u_i \in L\) and \(v_i \in R\) for all \(i\).
    In such case, we say that the root is \(u_1\), and the cycle is oriented from \(u_i\) to \(v_i\) and from \(v_i\) to \(u_{i+1}\).
    There are at most \((n/2)^{2k}\) possible such sequences, and, for each node in the sequence, we have at most \(\Delta^2\) choices for the incoming and outgoing half-edge.
    Overall, among all possible choices of such sequences and of half-edges, we have \((n/2)^{2k}\Delta^{4k}\) candidates.
    Fix one such candidate. For every \(i \in \{1,\ldots,k\}\), let \(a_i^+,a_i^- \in H_L\) be the half-edges of \(u_i\) used by the edges \(u_iv_i\) and \(u_iv_{i-1}\), respectively, where indices are taken modulo \(k\). Similarly, let \(b_i^-,b_i^+ \in H_R\) be the half-edges of \(v_i\) used by the edges \(u_iv_i\) and \(u_{i+1}v_i\), respectively.
    For a candidate representing a cycle, the \(2k\) half-edges \(a_1^+,a_1^-,\ldots,a_k^+,a_k^-\) must be pairwise distinct, and the same holds for the \(2k\) half-edges \(b_1^-,b_1^+,\ldots,b_k^-,b_k^+\).
    The candidate occurs precisely when the \(2k=r\) prescribed pairs
    \[
        \{a_i^+,b_i^-\}\quad\text{and}\quad \{a_{i+1}^-,b_i^+\}, \qquad i \in \{1,\ldots,k\},
    \]
    belong to the random perfect matching, where again the indices are taken modulo \(k\).
    The probability that this happens is exactly
    \begin{align*}
        \frac{1}{\frac{n\Delta}{2}\left(\frac{n\Delta}{2}-1\right)\left(\frac{n\Delta}{2}-2\right)\ldots\left(\frac{n\Delta}{2}-r + 1\right)}.
    \end{align*}  
    Hence, 
    \begin{align*}
        \expect{X_r} \le \frac{\left(\frac{n}{2}\right)^{2k} \Delta^{4k}}{\frac{n\Delta}{2}\left(\frac{n\Delta}{2}-1\right)\left(\frac{n\Delta}{2}-2\right)\ldots\left(\frac{n\Delta}{2}-r + 1\right)}.
    \end{align*}
    Note that \(r<g=O(\log n)\) and, hence, \(r \le n\Delta/4\) for sufficiently large \(n\).
    Hence, for all \(0 \le j \le r - 1\), we have
    \(n\Delta/2-j \ge n\Delta/4\), which we can replace in the inequality above, obtaining
    \begin{align*}
        \expect{X_r} & \le \left(\frac{n}{2}\right)^{2k} \Delta^{4k}\left(\frac{4}{n\Delta}\right)^{2k} \\
        & = (2\Delta)^{2k}.
    \end{align*}
    We distinguish two cases.
    If \(g=3\), then only cycles of length \(2\) contribute to \(X_{<g}\), and
    \[
        \expect{X_{<g}}=\expect{X_2}\le(2\Delta)^2=4\Delta^2\le4n^{1/2}.
    \]
    If \(g\ge4\), then \(g=\floor{\gamma\log_\Delta n}\), and the estimates above imply
    \[
        \expect{X_{<g}}\le g(2\Delta)^g.
    \]
    Moreover,
    \[
        (2\Delta)^g
        \le \Delta^{g\log_\Delta(2\Delta)}
        \le \Delta^{2g}
        \le n^{2\gamma},
    \]
    and hence \(\expect{X_{<g}}\le n^{3\gamma}\le n^{1/2}\) for all sufficiently large \(n\).
    Thus, in both cases, \(\expect{X_{<g}}=O(n^{1/2})\).
    By Markov's inequality,
    \begin{align*}
        \pr{X_{<g} \ge n^{2/3}} = O(n^{-1/6})=o(1),
    \end{align*}
    proving property 2.

    Since both properties 1 and 2 hold with probability \(1-o(1)\), a graph \(G''\) satisfying both properties must exist.
    In order to conclude the lemma, we must remove all short cycles by keeping the graph \(\Delta\)-regular.

    \paragraph{The final graph \(G\).}
    We proceed by applying an iterative edge-switching procedure.
    Initially, let \(F=G''\). As long as \(F\) contains a cycle of length strictly less than \(g\), take any such cycle \(C\) and an edge \(xy\) in it, where \(x \in L\) and \(y \in R\).
    We choose another edge \(uv\) of the current graph \(F\), with \(u \in L\) and \(v \in R\), such that
    \begin{enumerate}
        \item[(a)] The distance in \(F\) between \(\{u,v\}\) and \(C\) is at least \(g\).
        \item[(b)] The edge \(uv\) does not belong to any cycle of \(F\) of length strictly less than \(g\).
    \end{enumerate}
    We obtain the next current graph \(F'\) by deleting \(xy\) and \(uv\), and adding \(xv\) and \(yu\).
    This operation preserves bipartiteness and \(\Delta\)-regularity.
    We now show that such an edge \(uv\) always exists and that this operation strictly decreases the number of cycles of length less than \(g\).

    \paragraph{Existence of \(uv\).}
    Inductively, the current graph \(F\) contains at most \(n^{2/3}\) cycles of length less than \(g\): this holds initially by property 2, and the argument below shows that every switching strictly decreases their number.
    Since \(C\) has at most \(g-1\) nodes and the maximum degree is \(\Delta\), the number of nodes at distance less than \(g\) from \(C\) is at most
    \[
        g(1 + \Delta + \ldots + \Delta^{g-1}) \le g \Delta^g.
    \]
    The number of edges incident to these is at most \(g\Delta^{g + 1}\).
    Furthermore, the number of edges that belong to cycles of length less than \(g\) is at most \(g n^{2/3}\).
    Therefore, the number of edges that fail at least one of properties (a) and (b) is at most \(gn^{2/3}+g\Delta^{g+1}\).
    By the definition of \(g\),
    \[
        \Delta^g\le\max\{\Delta^3,n^\gamma\}
        \le\max\{n^{3/4},n^\gamma\}.
    \]
    Since \(g=O(\log n)\), it follows that
    \[
        gn^{2/3}+g\Delta^{g+1}=o(n\Delta).
    \]
    As the graph has \(n\Delta/2\) edges, for all sufficiently large \(n\) at least one edge remains after excluding all edges that fail (a) or (b).
    In particular, at every iteration there is at least one edge \(uv\) satisfying properties (a) and (b).

    \paragraph{Switching edges strictly decreases the number of short cycles.}
    Suppose that \(F'\) contains a new cycle \(D\) of length strictly less than \(g\), that is, a cycle that was not already present in \(F\).
    Then, \(D\) must contain at least one of the new edges \(xv\) and \(yu\).
    Suppose it contains \(xv\) but \emph{not} \(yu\).
    Removing \(xv\) from \(D\) gives an \(x\)-to-\(v\) path in \(F\) of length at most \(g-2\).
    Since \(x \in C\), this implies that the distance in \(F\) between \(v\) and \(C\) is less than \(g\), contradicting property (a).
    The same argument holds for the case where \(D\) contains only \(yu\) and not \(xv\).

    Suppose now that \(D\) contains both \(xv\) and \(yu\).
    Removing these two edges splits \(D\) into two paths, all of whose edges belong to \(F\).
    If these paths connect \(x\) to \(y\) and \(u\) to \(v\), then property (b) implies that the \(u\)-to-\(v\) path has length at least \(g-1\): otherwise, this path together with \(uv\) would be a cycle of \(F\) of length less than \(g\).
    This gives \(\abs{D} \ge g+1\), a contradiction.
    In the other case, the paths connect \(x\) to \(u\) and \(y\) to \(v\).
    Since \(x,y \in C\), property (a) implies that both paths have length at least \(g\), again contradicting \(\abs{D}<g\).

    Thus, the switching creates no new cycle of length less than \(g\).
    It deletes at least the selected cycle \(C\), because \(xy\) is removed, so the number of short cycles strictly decreases.
    The procedure therefore terminates after at most \(n^{2/3}\) iterations; let \(G\) be its final graph.
    If \(\gamma\log_\Delta n<4\), then
    \[
        g\ge3>\frac{\gamma}{2}\log_\Delta n.
    \]
    Otherwise,
    \[
        g\ge\floor{\gamma\log_\Delta n}
        \ge\frac{\gamma}{2}\log_\Delta n.
    \]
    Thus, \(G\) has the girth required in the statement after setting \(c=\gamma/2\).
    Moreover, the construction produces no loops, and the switching procedure removes every cycle of length \(2\). Therefore, the final graph \(G\) is simple.

    \paragraph{Every maximal matching in \(G\) is large.}
    Let \(t \le n^{2/3}\) be the number of switchings that we perform to obtain \(G\) from \(G''\).
    Since every switching removes two edges,
    there are at most \(2t\) edges that are present in \(G''\) that do not belong to \(G\).
    Let \(M\) be an arbitrary maximal matching in \(G\), and let \(U_L \subseteq L\) and \(U_R \subseteq R\) be the sets of unmatched nodes.
    We trivially have \(\abs{U_L} = \abs{U_R}\), and there is no edge between these two sets (otherwise we could increase the size of the maximal matching).
    This implies that if, in \(G''\), there were edges between \(U_L\) and \(U_R\), these must have been removed during edge-switching.
    For all such removed edges, we consider the subset \(U_L'\) where we remove the endpoints of these edges from \(U_L\).
    We have \(\abs{U_L'} \ge \abs{U_L}- 2t\).
    Trivially, there is no edge between \(U_L'\) and \(U_R\) in \(G''\).
    Consider a subset \(U_R'\subseteq U_R\) such that \(\abs{U_L'} = \abs{U_R'}\).
    Then, by property 1 in \(G''\), we have that \(\abs{U_L'} = \abs{U_R'} < C (n/(2\Delta)) \ln \Delta\).
    Hence,
    \begin{align*}
        \abs{U_L} - 2t \le C \frac{n \ln \Delta}{2\Delta},
    \end{align*}
    which implies that
    \begin{align*}
        \abs{U_L} \le C \frac{n \ln \Delta}{2\Delta} + 2t.
    \end{align*}
    Therefore, the total number of unmatched nodes is at most
    \begin{align*}
        \abs{U_L} + \abs{U_R} & \le 2C \frac{n \ln \Delta}{2\Delta} + 4t \\
        & \le 4C\frac{n \ln \Delta}{2\Delta},
    \end{align*}
    where the latter holds because \(t\le n^{2/3}=o(n\ln\Delta/\Delta)\) uniformly over \(3\le\Delta\le n^{1/4}\).
    After increasing the constant \(C\) by a factor of two, this is the bound claimed in the statement.
\end{proof}
\end{document}

%% file: packages.tex
\usepackage{amsmath,amsthm,thmtools,amsfonts,amssymb}
\usepackage{dsfont}
\usepackage{mathtools}
\usepackage{mathrsfs}
\usepackage{physics}
\usepackage{array}
\allowdisplaybreaks

\usepackage{csquotes}
\usepackage[
    backend=biber,
    style=alphabetic,
    sorting=nyt,
    minalphanames=3,
    maxbibnames=99
]{biblatex}
\usepackage[margin=1in]{geometry}
\usepackage{xspace}
\usepackage{xcolor}
\usepackage{enumitem}
\usepackage{microtype}

\usepackage{graphicx}
\usepackage{subcaption}
\usepackage{tikz}
\usetikzlibrary{calc}
\input{figures/tikz-styles.tex}

\usepackage{algorithm}
\usepackage{algpseudocode}

\usepackage{hyperref}

\hypersetup{
    colorlinks=true,
    linkcolor=black,
    citecolor=black,
    filecolor=black,
    urlcolor=black,
}
\usepackage[capitalize,noabbrev]{cleveref}

%% file: figures/tikz-styles.tex
\definecolor{figureblue}{RGB}{32,91,135}
\tikzset{
  formalism edge/.style={draw=black!65,line width=.65pt},
  formalism matched/.style={draw=figureblue,line width=1.5pt},
  formalism white/.style={circle,draw=black,line width=.7pt,
    fill=white,inner sep=0pt,minimum size=5.5pt},
  formalism black/.style={formalism white,fill=black},
  formalism label/.style={font=\footnotesize,fill=white,inner sep=1pt},
  formalism root/.style={draw=figureblue,line width=.6pt},
}

%% file: environments.tex
\newtheorem{theorem}{Theorem}[section]
\newtheorem{lemma}[theorem]{Lemma}

\newtheorem{observation}[theorem]{Observation}

\theoremstyle{definition}
\newtheorem{definition}[theorem]{Definition}

\theoremstyle{remark}

%% file: macros.tex
\DeclarePairedDelimiterXPP{\myO}[1]{O}{(}{)}{}{#1}
\DeclarePairedDelimiterXPP{\myTheta}[1]{\Theta}{(}{)}{}{#1}
\DeclarePairedDelimiterXPP{\myOmega}[1]{\Omega}{(}{)}{}{#1}
\DeclarePairedDelimiterXPP{\mylittleo}[1]{o}{(}{)}{}{#1}
\DeclarePairedDelimiterXPP{\mylittleomega}[1]{\omega}{(}{)}{}{#1}

\DeclarePairedDelimiterXPP{\pr}[1]{\Pr}{[}{]}{}{#1}
\DeclarePairedDelimiterXPP{\expect}[1]{\mathbb{E}}{[}{]}{}{#1}
\DeclarePairedDelimiterXPP{\variance}[1]{\operatorname{Var}}{[}{]}{}{#1}
\DeclarePairedDelimiter{\ceil}{\lceil}{\rceil}
\DeclarePairedDelimiter{\floor}{\lfloor}{\rfloor}

\makeatletter
\newcommand{\defineuppermathalphabet}[2]{%
    \@tfor\mathletter:=ABCDEFGHIJKLMNOPQRSTUVWXYZ\do{%
        \expandafter\edef\csname #1\mathletter\endcsname{%
            \noexpand#2{\mathletter}}}}
\newcommand{\definelowermathalphabet}[2]{%
    \@tfor\mathletter:=abcdefghijklmnopqrstuvwxyz\do{%
        \expandafter\edef\csname #1\mathletter\endcsname{%
            \noexpand#2{\mathletter}}}}

\defineuppermathalphabet{cal}{\mathcal}
\definelowermathalphabet{cal}{\mathcal}

\defineuppermathalphabet{bb}{\mathbb}

\defineuppermathalphabet{bf}{\mathbf}
\definelowermathalphabet{bf}{\mathbf}

\defineuppermathalphabet{scr}{\mathscr}

\defineuppermathalphabet{frak}{\mathfrak}
\definelowermathalphabet{frak}{\mathfrak}

\defineuppermathalphabet{sf}{\mathsf}
\definelowermathalphabet{sf}{\mathsf}

\defineuppermathalphabet{tt}{\mathtt}
\definelowermathalphabet{tt}{\mathtt}
\makeatother

\DeclareMathOperator{\re}{\mathcal R}
\newcommand{\constr}[1]{\mathcal{#1}}
\newcommand{\discr}{\mathcal{D}}
\newcommand{\dom}{D}

\newtheorem{openquestion}{Open Question}

%% file: trunk/notes.tex
\section{Introduction}
 
In the distributed setting, the network is represented by a graph in which nodes model computing entities and edges represent communication links. One of the most studied models of distributed computation is the LOCAL model. This is a synchronous message-passing model; computation proceeds in synchronous rounds, where in each round nodes exchange messages with their neighbors and then perform local computation. Neither the message size nor the computational power of individual nodes is restricted. Each node has a unique identifier and, when randomness is allowed, each node has access to an unbounded string of random bits.

When studying problems in this model, the focus is on communication, and the goal is to understand how many rounds of communication are sufficient, or necessary, in order to solve a graph problem.
The LOCAL model captures the \emph{locality} of distributed graph problems: how far must information travel through a network in order to solve a given problem? Understanding the locality of fundamental graph problems in the distributed setting has been at the center of attention since the field's beginnings. Over the years, researchers have studied the locality of graph problems as a function of both the number $n$ of nodes and the maximum degree $\Delta$ of the graph. 

Collective efforts have focused both on developing algorithms with tighter upper bounds and on proving stronger lower bounds. For example, the best upper bound for maximal matching,\footnote{The maximal matching problem asks for a subset of edges that share no endpoints (that is, an independent set of edges) such that adding any other edge to the set would violate this requirement. By \emph{best} here we mean the upper bound that minimizes the dependence on $n$, at the cost of a larger dependence on $\Delta$.} expressed as a function of $n$ and $\Delta$, has been known since 2001 \cite{panconesi-rizzi-2001-some-simple-distributed-algorithms}. However, the optimality of this algorithm was established only 18 years later \cite{balliu-brandt-etal-2021-lower-bounds-for-maximal}.
Part of the reason why this result took so long to establish is the fact that the LOCAL model is very powerful, making it quite challenging to prove lower bounds. Yet, it is also highly rewarding, as lower bounds for the LOCAL model directly carry over to weaker models of distributed computing. Another key reason was the lack of general-purpose lower bound techniques: developing broadly applicable techniques can indeed be difficult, but it is sometimes the best way to gain a more comprehensive and deeper understanding of problems of interest. 

For several natural graph problems, the first polylogarithmic randomized LOCAL lower bounds were established back in 2004 using the celebrated KMW construction \cite{kuhn-moscibroda-wattenhofer-2016-local-computation}. At the core of this technique are arguments based on the notion of \emph{indistinguishability}.
At a very high level, the KMW lower bounds are obtained by showing that, if the running time of the nodes is too small, then there are many nodes in the construction that have identical information about the graph and hence have the same output distribution. 
Such behavior is problematic for problems that require breaking symmetry between nodes or edges, such as maximal matching or maximal independent set, and thus leads to lower bounds for such problems.
%Indistinguishability-based lower bounds are very powerful because they hold also in stronger models such as quantum-LOCAL or the non-signaling model \cite{gavoille-kosowski-markiewicz-2009-what-can-be-observed,coiteux-roy-d-amore-etal-2024-no-distributed-quantum}.

As time has shown, obtaining lower bounds via indistinguishability arguments appears to be quite challenging for \emph{certain} symmetry-breaking problems of interest. In fact, apart from the KMW construction and its closely related variants \cite{coupette-lenzen-2021-a-breezing-proof-of-the-kmw-bound,balliu-ghaffari-etal-2022-node-and-edge-averaged}, there are no genuinely new constructions yielding lower bounds via indistinguishability. 
An important turning point in the development of lower bounds for symmetry-breaking problems in the LOCAL model (and, consequently, in our understanding of the distributed complexity of such problems) came when an old technique called \emph{round elimination} \cite{linial-1992-locality-in-distributed-graph-algorithms}  was extended into a \emph{general-purpose} technique.

\paragraph{Round elimination.} The successful development of round elimination into a general-purpose technique dates back to 2019 \cite{brandt-2019-an-automatic-speedup-theorem-for}.
Broadly speaking, this technique is based on the following idea. Assume that a problem $\Pi_0$ of interest can be solved in $T$ rounds. The goal is to construct a lower bound sequence of problems $\Pi_1, \Pi_2,\ldots, \Pi_T$ such that each problem $\Pi_i$ can be solved in $T-i$ rounds. In particular, $\Pi_T$ can be solved in $0$ rounds. If we then show that $\Pi_T$ \emph{cannot} be solved in $0$ rounds, we obtain a chain of contradictions and conclude that $\Pi_0$ cannot be solved in $T$ rounds.
Given a problem $\Pi_0$, the sequence of problems used to prove the lower bound can be constructed in a \emph{mechanical} way. That is, there exists a function $\re$ that takes as input a problem $\Pi$ and outputs a problem $\Pi'$ that, under some conditions, is exactly one round easier than $\Pi$. Unfortunately, often the description of $\Pi'$ is exponentially larger than that of $\Pi$, making it practically impossible to obtain a useful lower bound.
In some cases, however, the problems remain the same, i.e., $\re(\Pi)=\Pi$. Clearly, $\Pi$ cannot be exactly one round easier than itself, so this implies some kind of contradiction. The contradiction is on the conditions required to conclude that $\re(\Pi)$ is one round easier than $\Pi$, and we know that, if such conditions do not apply, then $\Pi$ is in some sense \emph{hard}.
If $\re(\Pi)=\Pi$, we say that $\Pi$ is a \emph{fixed point} under the round elimination framework. For example, the perfect matching problem (i.e., finding a matching in which every node is matched) is a fixed point, whereas \emph{maximal matching is not a fixed point}. 
The general-purpose aspect of the round elimination technique has been key for obtaining many breakthrough lower bound results and has substantially advanced our understanding of the complexity of several important graph problems \cite{balliu-boudier-etal-2024-tight-lower-bounds-in-the,balliu-brandt-etal-2020-classification-of-distributed,balliu-brandt-etal-2021-improved-distributed-lower,balliu-brandt-etal-2021-lower-bounds-for-maximal,balliu-brandt-etal-2023-distributed-maximal-matching,balliu-brandt-etal-2025-distributed-quantum-advantage,balliu-brandt-etal-2025-towards-fully-automatic,balliu-brandt-etal-2026-on-the-universality-of-round,balliu-brandt-olivetti-2022-distributed-lower-bounds,balliu-hirvonen-etal-2019-hardness-of-minimal-symmetry,brandt-olivetti-2020-truly-tight-in-delta-bounds-for,
brandt-fischer-etal-2016-a-lower-bound-for-the,
balliu-brandt-etal-2026-distributed-delta-coloring}.

The lower bounds in the LOCAL model obtained using the broadly applicable round elimination technique are tight for some problems. For other problems, however, these bounds are not tight.
In fact, if we consider randomized algorithms and focus on expressing the complexity of graph problems solely as a function of $n$, the KMW construction establishes, for many problems of interest, a stronger lower bound than the one obtainable via the current round elimination technique. For example, if we focus on randomized algorithms that solve the maximal matching problem in the LOCAL model, the two techniques give the following bounds.

\begin{itemize}
    \item KMW proves a lower bound of $\Omega\left(\sqrt{\frac{\log n}{\log \log n}}\right)$ rounds;
    \item The current round elimination technique gives a lower bound of $\Omega\left(\frac{\log\log n}{\log\log\log n}\right)$ rounds.
\end{itemize}

\paragraph{A new lower bound technique: self-reduction.} In 2025, Khoury and Schild introduced a new technique that they called round elimination via \emph{self-reduction} \cite{khoury-schild-2025-round-elimination-via-self-reduction}. Using this technique, they improved the KMW lower bound for maximal matching, showing that any randomized algorithm that solves the problem in the LOCAL model requires $\Omega\left(\sqrt{\log n}\right)$ rounds. It took 21 years to improve the lower bound for such a fundamental problem.
At a high level, the lower bound is proved as follows. Assume that we can produce a sufficiently large matching, which may not be maximal, in $T$ rounds in the LOCAL model. It is possible to show that, in $T-1$ rounds, we can produce a matching that still matches many nodes (in other words, the matching produced in one round less is not much smaller). Iterating this argument, starting from a sufficiently small value of $T$, would yield a $0$-round algorithm that produces a matching larger than is possible without communication. This gives a contradiction, establishing a lower bound of $T+1$ rounds.

Like the round elimination technique, self-reduction examines what happens when we eliminate one round. However, while maximal matching is not a fixed point under the round elimination framework, self-reduction does not change the problem studied when moving to one round less: it remains a matching problem. Thus, unlike classical round elimination, self-reduction keeps the problem itself unchanged, hence the name ``self-reduction''.

The self-reduction technique immediately attracted lots of attention because new techniques for proving lower bounds in the LOCAL model can pave the way for understanding the locality of many graph problems, especially if they can be developed into general-purpose techniques. Round elimination is a striking example of this phenomenon. The original formulation and analysis of the self-reduction technique were quite complex. Shortly after its publication, the technique was significantly simplified and then extended to establish a lower bound of $\Omega\left(\sqrt{\log_{1+b} n}\right)$ rounds for the more general maximal $b$-matching problem \cite{balliu-casagrande-etal-2026-new-hardness-results-for}.
While this was a valuable step towards better understanding self-reduction, the goal of obtaining a general-purpose technique remained far from being achieved.
Since then, several researchers (including the authors of this paper) have put considerable effort into understanding the power of this technique and tried to apply it to other symmetry-breaking problems, but these attempts have consistently encountered technical obstacles.

\paragraph{Limitations of current round elimination.}
Being able to obtain randomized lower bounds of $\omega(\log \log n)$ rounds directly via classical round elimination would represent a major advance. Indeed, since round elimination is a general-purpose technique, such a result could pave the way for proving randomized lower bounds of $\omega(\log \log n)$ rounds for many different problems.
However, efforts to achieve this goal have been unsuccessful so far, suggesting that this may be a limitation of the technique itself. In fact, Khoury and Schild explicitly express this intuition in their paper \cite{khoury-schild-2025-round-elimination-via-self-reduction}.

\begin{quotation}
\noindent``While this approach [round elimination] has been very successful for several problems, including leading to the state-of-the-art lower bounds for deterministic algorithms, its performance is limited as a function of $n$ for randomized ones.''
\end{quotation}
The fundamental question of whether it is possible at all to prove randomized lower bounds of $\omega(\log \log n)$ rounds via classical round elimination has remained open since 2019 and was explicitly posed in 2022 as Open Problem 6 in \cite{balliu-brandt-etal-2022-distributed-delta-coloring,balliu-brandt-etal-2026-distributed-delta-coloring}.

\subsection{Our Contribution in a Nutshell}

In this paper, we provide a deep understanding of self-reduction and we turn it into a general-purpose technique. More precisely, we show that the self-reduction technique is a special case of the classical round elimination technique. To achieve this result, we first characterize a vast class of problems for which the self-reduction technique works. We then improve the round elimination technique to prove, for this entire class of problems, a stronger lower bound than what was previously obtainable via round elimination. 

To establish our results, we advance our understanding of round elimination in two directions. First, we introduce a new way to measure an algorithm's error (that is, a new way to count the incorrect outputs produced by the algorithm), enabling a more fine-grained analysis of its failure probability. This alone, however, is not sufficient. We also introduce a new way to construct a $(T-1)$-round algorithm from a $T$-round one, significantly different from the usual way.

However, the class of problems for which we obtain the above-mentioned results includes only rigid ``global'' problems, such as perfect matching and 2-vertex coloring. It is already known that these problems, if solvable at all, require $\Omega(D)$ rounds in the LOCAL model, where $D$ is the diameter of the graph. We also know that any solvable problem in the LOCAL model can be solved in $O(D)$ rounds.
Hence, it is natural to ask:
What is the value of obtaining a stronger lower bound via round elimination for these problems when we already know an asymptotically tight lower bound?
What is the relationship between rigid global problems and local problems? More specifically, how can a stronger lower bound obtained via round elimination for perfect matching imply a stronger lower bound for the much easier problem of maximal matching? 

The answers to the above questions lie in a fundamental difference: an $\Omega(D)$ lower bound applies only when the algorithm is required to fail with very small probability; in contrast, our technique yields lower bounds \emph{as a function of the target error} of the algorithm, with a dependence that \emph{scales well} with the error parameter.
Note that standard round elimination also gives lower bounds that depend on the target failure probability, but this dependence is substantially weaker than the one we obtain in this paper.

Let us compare the lower bounds that we get using our technique with those obtained via standard round elimination. Let $p$ be the target local failure probability. With standard round elimination, we obtain lower bounds of 
\[
\Omega\left(\min\left\{\log_\Delta \log \frac{1}{p},\log_\Delta n\right\}\right) \mbox{ rounds,}
\]
whereas our new error measure and algorithm definitions yield a lower bound of

\[
\Omega\left(\min\left\{\log \frac{1}{p},\log_\Delta n\right\}\right) \mbox{ rounds.}
\]
For the moment, let us ignore the $\log_\Delta n$ term, which becomes relevant only when $p$ is very small.
Our lower bounds represent more than an exponential improvement over the standard ones. Indeed, an exponential improvement alone would improve the first term only to $\log_\Delta \frac{1}{p}$, whereas our lower bound additionally removes the dependence on $\Delta$ from the base of the logarithm.
It is precisely the \emph{absence} of this dependence on $\Delta$, together with the exponential improvement, that enables us to prove lower bounds for problems \emph{outside} the class of rigid problems mentioned above.

As an example, let us focus on maximal matching (a problem outside the class) and perfect matching (a problem in the class). First, we observe that there exist families of graphs in which \emph{any} maximal matching leaves at most an $O(\log\Delta/\Delta)$ fraction of the nodes unmatched. On these graphs, we can transform any algorithm that solves maximal matching into an algorithm that solves perfect matching with sufficiently small error (i.e., an $O(\log\Delta/\Delta)$ fraction of nodes produce an incorrect output),  without increasing the running time. Our refined analysis, together with our new definitions, immediately implies that any randomized algorithm that solves perfect matching with sufficiently small error requires $\Omega\left(\sqrt{\log n}\right)$ rounds in the LOCAL model. This implies the same lower bound for maximal matching.

As a consequence of our results, all previously known lower bounds obtained via self-reduction follow as simple corollaries. Moreover, we prove new lower bounds for a range of graph problems, namely maximal matching on \emph{bipartite} regular $2$-colored graphs, $1/k$-integral matching, and maximal $H$-packing. In the latter problem, given a constant-size graph $H$, the goal is to find a maximal collection of vertex-disjoint copies of $H$ in the input graph. As with randomized lower bounds obtained via round elimination in general, our lower bounds hold even when nodes have access to \emph{shared randomness}.

In summary, we establish \emph{the first} $\Omega\left(\sqrt{\log n}\right)$ lower bounds directly via standard round elimination, showing that classical round elimination can yield stronger lower bounds than those obtained via KMW, even when the complexity is expressed solely as a function of $n$. This demonstrates that previous attempts to prove stronger lower bounds via round elimination failed because of insufficiently precise analyses and definitions, rather than because of an inherent limitation of the technique itself. This also resolves Open Problem 6 in \cite{balliu-brandt-etal-2026-distributed-delta-coloring} in the affirmative.

To fully describe our results and the new technical insights behind them, we must first explain round elimination (and fixed points) in more detail, define the class of problems for which we obtain stronger lower bounds via round elimination, and give some useful definitions. We would like to point out that the results obtained in this paper would not have been possible without the beautiful ideas and techniques of Khoury and Schild \cite{khoury-schild-2025-round-elimination-via-self-reduction}.

\section{High-Level Ideas, Results, and Technical Insights}
In this paper, we show that the result of Khoury and Schild \cite{khoury-schild-2025-round-elimination-via-self-reduction} can be obtained via standard round elimination, and we use the resulting insights to characterize the class of problems to which this approach applies. We start by explaining the round elimination technique in more detail.

\subsection{A Summary of Round Elimination}

The round elimination technique does not work directly in the LOCAL model, but rather in the deterministic version of a weaker model called the \emph{port-numbering} model. In this setting, the computational power of each node and the message size remain unrestricted. However, nodes have neither unique identifiers nor access to randomness. At each node $v$, the incident edges are assigned distinct port numbers from $1$ to $\deg(v)$, in some arbitrary order, where $\deg(v)$ is the degree that $v$ has in the graph. Round elimination works in this setting as follows. 

There exists a function $\re$ that takes as input a problem $\Pi$ defined in the so-called black-white formalism, and outputs a problem $\re(\Pi)$ that is still defined in the black-white formalism (the details of this formalism are not important at this point; we will however define it shortly in \cref{def:white-black}).
This function guarantees that, if $\Pi$ has complexity $T$ in the deterministic port-numbering model, and $T$ is sufficiently small compared to the girth of the graph, then $\re(\Pi)$ has complexity $\max\{T-1,0\}$.
Some problems satisfy an interesting property: $\re(\Pi) = \Pi$.
These problems are called \emph{fixed points}.
Note that, if $\Pi$ cannot be solved in $0$ rounds, this is a clear contradiction, and we can conclude that $T$ cannot be sufficiently small compared to the girth of the graph. By picking suitable graph families, we obtain a lower bound of $\Omega(\log_\Delta n)$ rounds (recall that $n$ denotes the number of nodes and $\Delta$ the maximum degree of the graph). 

In summary, if a nontrivial problem $\Pi$ satisfies $\re(\Pi) = \Pi$, we immediately obtain a lower bound of $\Omega(\log_\Delta n)$ rounds in the deterministic port-numbering model. However, what we really want are lower bounds in the much stronger randomized LOCAL model. For this purpose, we can use existing black-box \emph{lifting} theorems that lift such lower bounds to the randomized LOCAL model.
In fact, the following is known. 
\begin{theorem}[\cite{balliu-brandt-etal-2026-distributed-delta-coloring,balliu-brandt-etal-2026-on-the-universality-of-round}]
    Let $\Pi$ be a problem described in the black-white formalism using $L$ labels. Assume that $\re(\Pi) = \Pi$ and that $\Pi$ cannot be solved in $0$ rounds in the deterministic port-numbering model. Then, $\Pi$ requires $\Omega(\log_\Delta n - \log_\Delta \log L)$ rounds in the deterministic LOCAL model, and $\Omega(\log_\Delta \log n - \log_\Delta \log L)$ rounds in the randomized LOCAL model, for algorithms with failure probability at most $1/n$.
\end{theorem}
For problems where $L$ is at most $2^{O(\Delta)}$, we directly get a randomized LOCAL lower bound of $\Omega(\log_\Delta \log n)$ rounds. As a side note, one can make such black-box lifting theorems parametric in the desired local failure probability.
Indeed, the following is also known.
\begin{theorem}[\cite{MausRSS2026}]
    Let $\Pi$ be a problem that satisfies $\re(\Pi) = \Pi$ and that cannot be solved in $0$ rounds in the deterministic port-numbering model. Also, assume that the number of labels of $\Pi$ is in $2^{O(\Delta)}$. Then, $\Pi$ requires $\Omega\left(\min\left\{\log_\Delta n,\log_\Delta \log \frac{1}{p}\right\}\right) - O(1)$ rounds in the randomized LOCAL model, for algorithms with local failure probability at most $p$.
\end{theorem}

The above black-box lifting theorems are proved by showing the following statement, and applying it recursively.
\begin{lemma}[\cite{balliu-brandt-etal-2020-classification-of-distributed,grunau-rozhon-brandt-2022-the-landscape-of-distributed}, informal]
    Suppose $\Pi$ can be solved in $T$ rounds with local failure probability $p$. Then, $\re(\Pi)$ can be solved in $T-1$ rounds with local failure probability at most $O\left(p^{\frac{1}{\Delta+1}}\right)$.
\end{lemma}
In other words, the lemma shows that, for randomized algorithms, we can save one round at the cost of increasing the local failure probability \emph{polynomially}.

An important observation is that applying these black-box lifting theorems within the round elimination framework cannot yield lower bounds stronger than $\Omega(\log_\Delta \log n)$ rounds. More generally, for algorithms with local failure probability at most $p$, these theorems cannot establish lower bounds stronger than $\Omega\left(\min\left\{\log_\Delta n,\log_\Delta \log\frac{1}{p}\right\}\right)$ rounds.
For some problems, such as sinkless orientation, this lower bound is tight. As a consequence, these black-box lifting theorems \emph{cannot be improved} without restricting the class of problems: there cannot be a lifting theorem that applies to all problems satisfying $\re(\Pi)=\Pi$ and that gives an $\omega(\log_\Delta \log n)$ lower bound. Hence, a stronger lifting theorem requires additional restrictions on the class of problems.

\subsection{Problem Family}
To describe the family of problems we consider, we first need to provide the definition of the \emph{black-white formalism} used in the round elimination framework.

\begin{definition}[Problems in the black-white formalism]\label{def:white-black}
    A problem $\Pi$ in the black-white formalism is a tuple $(\Sigma,\Delta,\delta, \constr{C}_W,\constr{C}_B)$ where:
    \begin{itemize}
        \item $\Sigma$ is a finite set of labels.
        \item $\constr{C}_W$ is a set of multisets, each of size $\Delta$, of labels from $\Sigma$.
        \item $\constr{C}_B$ is a set of multisets, each of size $\delta$, of labels from $\Sigma$.
    \end{itemize}
\end{definition}
Elements of $\constr{C}_W$ and $\constr{C}_B$ are called \emph{configurations}.
Solving a problem $\Pi = (\Sigma,\Delta,\delta, \constr{C}_W,\constr{C}_B)$ in a $(\Delta,\delta)$-biregular $2$-colored graph means assigning a label from $\Sigma$ to each edge such that:
\begin{itemize}
    \item For each white node $v$, the multiset of labels assigned to the edges incident to $v$ forms a configuration present in $\constr{C}_W$.
    \item For each black node $v$, the multiset of labels assigned to the edges incident to $v$ forms a configuration present in $\constr{C}_B$.
\end{itemize}
Throughout the paper, we assume that the set of labels appearing in $\constr{C}_W$ and the set of labels appearing in $\constr{C}_B$ are the same (otherwise there would be configurations that cannot be used at all, and we normalize the problem by removing them), and are exactly $\Sigma$. Moreover, we assume that $\constr{C}_W$ and $\constr{C}_B$ are non-empty.

Note that, with this formalism, it is possible to define problems on standard non-$2$-colored graphs: nodes become white nodes, then we put a black node in the middle of each edge, and we require to output a label for each node-edge pair. This variant, which is essentially the restriction to the case $\delta=2$, is called \emph{node-edge formalism}. 

\Cref{fig:perfect-matching-formalisms} illustrates the perfect matching problem in both formalisms.
\input{figures/perfect-matching.tex}

\paragraph{Family of problems.}
Our improved black-box lifting theorem applies to all problems in the black-white formalism in which any two distinct configurations have edit distance at least $2$.
If, in addition, each configuration has a so-called \emph{dominant} label, we obtain stronger lower bounds. We now define these concepts formally.

\begin{definition}[Distance between configurations]
    Let $C$ be a configuration, and let $m_C(\ell)$ denote the multiplicity of $\ell$ in $C$. For two configurations $C_1$ and $C_2$ of the same size, we define their distance as
\[
    d(C_1,C_2)=\frac{1}{2}\sum_{\ell} |m_{C_1}(\ell) - m_{C_2}(\ell)|,
\]
that is, the minimum number of label substitutions needed to transform one configuration into the other, or, equivalently, the edit distance between $C_1$ and $C_2$.
\end{definition}

We extend this notion to define the distance between a configuration and a set of configurations in the standard way.
\begin{definition}[Distance from a set]
    Let $\constr{C}$ be a nonempty set of configurations, and let $C_1$ be a configuration. Then, the distance from $C_1$ to $\constr{C}$ is defined as 
    \[
    d(C_1,\constr{C}) = \min_{C_2 \in \constr{C}} d(C_1,C_2).
    \]
\end{definition}

\begin{definition}[$2$-separated problems]
    Let $\Pi = (\Sigma,\Delta,\delta,\constr{C}_W,\constr{C}_B)$ be a problem in the black-white formalism. We say that $\Pi$ is \emph{$2$-separated} if the following holds.
    \begin{itemize}
        \item Let $C_1,C_2 \in \constr{C}_W$. If $C_1 \neq C_2$, then $d(C_1,C_2) \ge 2$.
        \item Let $C_1,C_2 \in \constr{C}_B$. If $C_1 \neq C_2$, then $d(C_1,C_2) \ge 2$.
    \end{itemize} 
\end{definition}
Perfect matching, $2$-vertex coloring, even orientations are examples of $2$-separated problems. Instead, maximal matching, maximal independent set, sinkless orientation are examples of problems that are not $2$-separated.

\begin{definition}[$b$-dominant problems]
    Let $b\ge 0$ be an integer. Let $\Pi = (\Sigma,\Delta,\delta,\constr{C}_W,\constr{C}_B)$ be a problem in the black-white formalism. We say that $\Pi$ is \emph{$b$-dominant} if the following holds:
    \begin{itemize}
        \item Let $C \in \constr{C}_W$. There exists a label $\ell$ such that $m_C(\ell) \ge \Delta - b$.
        \item Let $C \in \constr{C}_B$. There exists a label $\ell$ such that $m_C(\ell) \ge \delta - b$.
    \end{itemize} 
\end{definition}
Observe that any problem in the black-white formalism is $\max\{\Delta,\delta\}$-dominant.

\subsection{A New Measure of Error}
Typically, the local failure probability of a node is defined as the probability that the node produces an invalid configuration, and the global failure probability is defined as the probability that at least one node fails.
Using this standard notion of failure (which is the one used in the existing lifting theorems) \emph{we cannot} achieve our results. In fact, one of our main contributions is proposing a novel and more fine-grained notion of error.
In order to distinguish the two notions, we use the term \emph{failure} for the standard one, and \emph{error} for the one we introduce.

\begin{definition}[Error of an algorithm]\label{def:error}
    Let $\Pi = (\Sigma,\Delta,\delta,\constr{C}_W,\constr{C}_B)$ be a problem in the black-white formalism. 
    Let $G = (W,B,E)$ be a $(\Delta,\delta)$-biregular $2$-colored graph, and let $f : E \rightarrow \Sigma$ be a label assignment to the edges of $G$. For a node $v$, let $c(v)$ be the configuration induced by the assignment of $f$ to the edges incident to $v$.
    The \emph{error} of $f$ w.r.t.\ $\Pi$ is defined as:
    \[
    \sum_{v\in W}d\bigl(c(v),\constr{C}_W\bigr) + \sum_{v \in B}d\bigl(c(v),\constr{C}_B\bigr)
    \]
    We call the first term \emph{white error} and the second term \emph{black error}.
    By \emph{expected error} of an algorithm $\mathcal{A}$ on $G$ we denote the expected error taken over all possible assignments of random bits  to the nodes of $G$.
\end{definition}
In other words, if a node fails, we do not simply count it as a failed node; instead, we measure its error by the minimum number of labels around it that must be changed to correct its output. Note that, however, this measure does not take into account any new errors that these changes around a node might introduce elsewhere.

\subsection{A New Lifting Theorem for Round Elimination}
We are now ready to state our main result about round elimination.
\begin{theorem}[Informal]\label{thm:newlifting-informal}
Let $\Pi$ be a $b$-dominant $2$-separated problem. 
Let $\varepsilon_0$ be a lower bound on the expected error of any $0$-round algorithm for $\Pi$.
Then, any algorithm for $\Pi$ with error at most $\varepsilon$ requires $\Omega\left(\min\left\{\log_\Delta n, \log_{2+b} \varepsilon_0/\varepsilon\right\}\right)$ rounds in the randomized LOCAL model.
\end{theorem}
We emphasize that \Cref{thm:newlifting-informal} is \emph{not} obtained by merely refining the analysis of the existing black-box lifting theorems in the case of $2$-separated problems and under the new notion of failure probability. To achieve it, we introduce a \emph{different} way to define the $(T-1)$-round algorithm.
In the following, we observe some consequences of our result. 

For ease of presentation, let us first simplify a bit the setting. Let us temporarily set aside the distinction between \emph{error} and \emph{failure} and assume that \cref{thm:newlifting-informal} talks only about failure probability (even though this is not quite precise).
Assume also that $\Delta \ge \delta$. Notice that, under this assumption, any problem in the black-white formalism is $\Delta$-dominant.
We observe the following.
\begin{itemize}
    \item Let $p$ be the desired local failure probability. The standard black-box lifting theorems give randomized lower bounds of $\Omega\left(\min\left\{\log_\Delta n, \log_\Delta \log \frac{1}{p}\right\}\right)$ rounds. As discussed, this bound cannot be improved without restricting the class of problems to which the theorem applies.
    \item We show that, for $2$-separated problems, we can improve this bound \emph{exponentially}, to $\Omega\left(\min\left\{\log_\Delta n, \log_\Delta \frac{1}{p}\right\}\right)$ rounds.
    \item If the problem is also $O(1)$-dominant, we obtain an even stronger result, establishing a lower bound of $\Omega\left(\min\left\{\log_\Delta n, \log \frac{1}{p}\right\}\right)$ rounds.
\end{itemize}
While generalizing the lower bound to obtain a better dependence on $p$ may seem like a small detail, we note that such an improvement has important consequences: in fact, it enables us to prove lower bounds for problems that are not even $2$-separated!

\subsection{Applications}
Now that we have provided formal definitions, we revisit the maximal matching example in more detail. Maximal matching, when encoded in the black-white formalism, is \emph{not} $2$-separated. However, \emph{perfect} matching is.
It is known that there are families of regular graphs with girth $\Omega(\log_\Delta n)$ and independence number $O(n\log\Delta/\Delta)$.
Since the unmatched nodes in any maximal matching form an independent set, every maximal matching on such graphs leaves at most an $O(\log\Delta/\Delta)$ fraction of the nodes unmatched. 
We transform any algorithm that solves maximal matching on these graphs into an algorithm that, in the same running time, solves \emph{perfect matching} with sufficiently small errors. The transformation is simple: each unmatched node arbitrarily chooses an incident edge and marks it as \emph{matched}.
Notice that the resulting algorithm never makes errors at nodes; errors occur only on edges, and the total number of errors is $O(n\log\Delta/\Delta)$. Hence, the expected error is $\varepsilon=O(n\log\Delta/\Delta)$. Also, it is easy to see that any $0$-round algorithm for perfect matching makes $\Omega(n)$ errors in expectation.
Since perfect matching is a $1$-dominant $2$-separated problem, we directly obtain a lower bound of $\Omega\left(\min\left\{\log_\Delta n, \log(\Delta/\log\Delta)\right\}\right) = \Omega\left(\min\left\{\log_\Delta n, \log\Delta\right\}\right)$ rounds for maximal matching.
In other words, we can use a simple graph theoretic argument, combined with our black-box lifting theorem, and obtain lower bounds for local problems.
While the example was for algorithms that never fail, even if we consider Monte Carlo algorithms with global failure probability at most $1/n$, we can alter the solution (without spending any additional round) so that every node is happy and errors are on edges, while introducing $O(1)$ additional expected errors, and hence the same idea goes through.

We apply our technique to two families of problems.
The first generalizes standard maximal matching to a fractional setting, where the goal is to assign each edge a weight in $[0,1]$ such that the sum of the weights of the edges incident to each node is at most $1$, and no edge weight can be increased without violating this constraint. A node is considered matched if the sum of its incident edge weights is exactly $1$. We restrict our attention to the case where outputs are restricted to integer multiples of $1/k$, for some constant $k$ independent of $n$ and $\Delta$, and prove the following theorem.

\begin{restatable}{theorem}{kintegral}\label{thm:kintegral}
Let $k$ be a positive integer fixed independently of $\Delta$ and $n$. Maximal $\frac{1}{k}$-integral matching requires $\Omega\left(\min\left\{\log_\Delta n, \log \Delta\right\}\right)$ rounds and $\Omega\left(\sqrt{\log n}\right)$ rounds in the randomized LOCAL model, for any algorithm with failure probability at most $1/n$. These lower bounds hold even on $\Delta$-regular graphs of girth $\Omega(\log_\Delta n)$.
\end{restatable}

The second family of problems generalizes matching in a totally different way.
We can see matching as the problem of packing a maximal set of copies of the graph consisting of two nodes that share an edge. A generalization of this view is maximal $H$-packing, where we use different graph patterns.
\begin{restatable}{theorem}{hpacking}\label{thm:hpacking}
    Let $H$ be a tree with at least $2$ nodes fixed independently of $\Delta$ and $n$. The maximal $H$-packing problem requires $\Omega\left(\min\left\{\log_\Delta n, \log \Delta\right\}\right)$ rounds and $\Omega\left(\sqrt{\log n}\right)$ rounds in the randomized LOCAL model, for any algorithm with failure probability at most $1/n$. These lower bounds hold even on $\Delta$-regular graphs of girth $\Omega(\log_\Delta n)$.
\end{restatable}

Maximal matching on $2$-colored graphs has been extensively studied. The best randomized lower bound as a function of $n$ is given by the KMW construction \cite{kuhn-moscibroda-wattenhofer-2016-local-computation} (in fact, the lower bound of Khoury and Schild \cite{khoury-schild-2025-round-elimination-via-self-reduction} does not directly apply in such graphs). 
We show that a simple graph-theoretic argument, combined with our new lifting theorem, yields an improved lower bound for maximal matching also on bipartite $2$-colored graphs.

\begin{restatable}{theorem}{bmm}\label{thm:bmm}
The maximal matching problem, on bipartite $\Delta$-regular $2$-colored graphs of girth $\Omega(\log_\Delta n)$, requires $\Omega\left(\min\left\{\log_\Delta n, \log \Delta\right\}\right)$ rounds and $\Omega\left(\sqrt{\log n}\right)$ rounds in the randomized LOCAL model, for any algorithm with failure probability at most $1/n$.
\end{restatable}

\paragraph{Limitations.}
In some sense, our new black-box lifting theorem \emph{cannot be improved}.
More precisely, we show that requiring problems to be $2$-separated is, in a certain sense, \emph{necessary}.
However, we do not claim that $b$-dominance is necessary. Nevertheless, we show that there are problems that are not $o(\Delta)$-dominant and that admit an algorithm with complexity $O\left(\log_\Delta \frac{\varepsilon_0}{\varepsilon}\right)$. This implies that there cannot exist a black-box lifting theorem for all $2$-separated problems that gives a lower bound of $\Omega\left(\min\left\{\log_\Delta n, \log \frac{\varepsilon_0}{\varepsilon}\right\}\right)$ rounds without imposing some additional restriction. We show that $O(1)$-dominance is sufficient, although other notions may be sufficient as well. We formalize these limitations in the following theorems.

\begin{theorem}[Informal]\label{thm:limit1-informal}
    Without $2$-separation, our round-elimination analysis cannot guarantee, in general, that saving one round increases the error by only a multiplicative factor depending on $\Delta$ and $\delta$.
\end{theorem}

\begin{restatable}{theorem}{limittwo}\label{thm:limit2}
    There exists a problem $\Pi$ that is $2$-separated, not $o(\Delta)$-dominant, and that can be solved in  $T \in O\left(1+\log_\Delta \frac{\varepsilon_0}{\varepsilon}\right)$ rounds in the randomized LOCAL model in graphs of girth $\Omega(\log_\Delta n)$ with target error at most $\varepsilon$, where $\varepsilon_0 \in \Omega(n)$ is a lower bound on the expected error of any $0$-round algorithm for the problem that has $0$ white error, assuming $\varepsilon$ is chosen such that $T \in o(\log_\Delta n)$.
\end{restatable}

\section{Road Map}
The paper is structured as follows.
In \Cref{sec:preliminaries}, we provide some basic notions, and we define the function $\re$.
In \Cref{sec:asr}, we prove \Cref{thm:newlifting-informal}, our new black-box lifting theorem for round elimination.
In \Cref{sec:applications}, we provide applications. We first provide a generic helper theorem, that connects local problems to global problems. Then, we apply it to prove \Cref{thm:kintegral} (the lower bound for $\frac{1}{k}$-integral matching), \Cref{thm:hpacking} (the lower bound for maximal $H$-packing), and \Cref{thm:bmm} (the lower bound for bipartite maximal matching).
In \Cref{sec:optimality}, we prove \Cref{thm:limit1-informal} ($2$-separation is required) and   \Cref{thm:limit2} ($2$-separation is not sufficient to improve $\log_\Delta \varepsilon_0 / \varepsilon$ to $\log \varepsilon_0 / \varepsilon$).
We conclude in \Cref{sec:open} with some open questions.

\section{Preliminaries}\label{sec:preliminaries}

\paragraph{LOCAL model.}
The LOCAL model of distributed computation is a synchronous message-passing model, first introduced by \textcite{linial-1992-locality-in-distributed-graph-algorithms}.
Typically, it is assumed that each node has an ID that is unique in the entire graph, and that nodes know beforehand their own ID, their degree in the graph, the maximum degree $\Delta$, and the total number $n$ of nodes in the graph.\footnote{While the assumptions regarding knowledge of global parameters such as $n$ and $\Delta$ may seem unnatural, they can only strengthen our lower-bound results.}

The computation proceeds in synchronous rounds. In each round, nodes exchange (possibly different) messages with their neighbors and perform local computation. The size of the messages and the amount of local computation performed by each node are not restricted. The runtime of an algorithm solving a graph problem in the LOCAL model is measured as the number of communication rounds needed for all nodes to produce their output.

In the randomized version of the LOCAL model, nodes have access to a string of random bits of unrestricted size.
As in the case of the port numbering model, nodes also have  access to \emph{ports}, i.e., an algorithm running on $v$ can e.g.\ refer to its $i$th incident edge. However, we assume that these ports are not given adversarially: they are assigned by the nodes using their own random bits (this is w.l.o.g., as any adversarial assignment could be replaced by a random assignment by the algorithm at the beginning of the computation).

Our lower bounds apply to Monte Carlo algorithms. Any $T$-round Monte Carlo algorithm that solves a graph problem in the randomized LOCAL model must produce a solution within $T$ rounds, and the solution must be correct with high probability, i.e., with probability at least $1 - 1/n$.

For technical reasons, we assume that nodes are not given IDs in the randomized LOCAL model. Note that this does not weaken the model (or our results), as nodes can always generate unique IDs with failure probability at most $1/n^c$, for any arbitrary $c >0$, without communication, and hence this does not change the failure probability by much.

\paragraph{Round elimination.}
Let $\Pi = (\Sigma,\Delta,\delta,\constr{C}_W,\constr{C}_B)$ be a problem in the black-white formalism. The function $\re(\Pi)$ gives a new problem $\Pi' = (\Sigma',\delta,\Delta,\constr{C}'_B,\constr{C}'_W)$ defined as follows.
\begin{itemize}
    \item Let $\constr{C}'$ be the set of multisets $\{L_1,\ldots,L_\delta\}$ satisfying that, for all $i$, $L_i \subseteq \Sigma$ and $L_i \neq \{\}$, and that $\forall \ell_1 \in L_1, \ldots, \ell_\delta \in L_\delta$, it holds that $\{\ell_1,\ldots,\ell_\delta\} \in \constr{C}_B$. Let $\constr{C}'_B$ be the subset of $\constr{C}'$ containing only \emph{maximal} configurations, that is, all configurations $C=\{L_1,\ldots,L_\delta\} \in \constr{C}'$ such that there is no other configuration $C' =\{L'_1,\ldots,L'_\delta\} \in \constr{C}'$ such that $C' \neq C$ and there exists a bijection $\phi$ such that, for all $i$, $L_i \subseteq L'_{\phi(i)}$. 
    \item $\Sigma'$ is the set of sets appearing in at least one configuration in $\constr{C}'_B$.
    \item Let $\constr{C}'_W$ be the set of multisets $\{L_1,\ldots,L_\Delta\}$ satisfying that, for all $i$, $L_i \in \Sigma'$, and that $\exists \ell_1 \in L_1, \ldots, \ell_\Delta \in L_\Delta$ such that $\{\ell_1,\ldots,\ell_\Delta\} \in \constr{C}_W$.
\end{itemize}
It is easy to observe that, for all $2$-separated problems, $\re\bigl(\re(\Pi)\bigr) = \Pi$, where we say that two problems are equal if they are the same problem up to label renaming.
\begin{observation}\label{obs:2sep-is-fixedpoint}
  Let $\Pi = (\Sigma,\Delta,\delta,\constr{C}_W,\constr{C}_B)$ be a $2$-separated problem, and let $\Pi' = \re(\Pi)$. Then, up to renaming of labels, $\Pi' = (\Sigma,\delta,\Delta,\constr{C}_B,\constr{C}_W)$.
\end{observation}
\begin{proof}
  Let us compute $\Pi' = (\Sigma',\delta,\Delta,\constr{C}'_B,\constr{C}'_W)$. Observe that the configurations present in $\constr{C}'_B$ are exactly the configurations $\{\{\ell_1\},\ldots,\{\ell_\delta\}\}$ such that $\{\ell_1,\ldots,\ell_\delta\} \in \constr{C}_B$, since the $2$-separated property implies that, by adding any label to any other set, we obtain a configuration (of sets) where we could pick an invalid configuration (of the original problem).
  Hence, $\Sigma' = \{\{\ell\} \mid \ell \in \Sigma\}$, and by renaming $\{\ell\}$ to $\ell$ we obtain that $\constr{C}'_B = \constr{C}_B$.
  Let us now compute $\constr{C}'_W$. It contains exactly those configurations with labels from $\Sigma'$ that satisfy an existential quantifier. Since each set in $\Sigma'$ is a singleton, this implies that $\constr{C}'_W$ contains exactly those configurations $\{\{\ell_1\},\ldots,\{\ell_\Delta\}\}$ such that $\{\ell_1,\ldots,\ell_\Delta\} \in \constr{C}_W$. Hence, by renaming $\{\ell\}$ to $\ell$ we obtain that $\constr{C}'_W = \constr{C}_W$.
\end{proof}

\paragraph{White and black algorithms.}
Let $\Pi = (\Sigma,\Delta,\delta,\constr{C}_W,\constr{C}_B)$ be a problem in the black-white formalism.
Throughout the paper, we will only consider $\Pi$ in the context of $(\Delta,\delta)$-biregular graphs where nodes are $2$-colored. We now define the notion of algorithms in this setting. Let \emph{white} be the color of the nodes that need to satisfy the constraint $\constr{C}_W$, and let \emph{black} be the color of the nodes that need to satisfy the constraint $\constr{C}_B$. A $T$-round algorithm $A$ for $\Pi$ is a function that takes as input the $T$-radius neighborhood of a white node (which has degree $\Delta$) and produces an output for each of its incident edges. We say that white nodes are \emph{active} and black nodes are \emph{passive}, in the sense that white nodes are the ones actually producing the outputs. 

Observe that an algorithm for $\re(\Pi) = (\Sigma',\delta,\Delta,\constr{C}'_B,\constr{C}'_W)$ is now a mapping from the neighborhoods of the nodes that need to satisfy $\constr{C}'_B$ (which have degree $\delta$) to outputs.
In other words, the roles of active and passive are now swapped. As a simple explanatory case, consider the setting where white nodes have degree $\Delta$ and black nodes have degree $2$.  The round elimination statement described in the introduction essentially says that, if there is a $T$-round algorithm for $\Pi$ where nodes of degree $\Delta$ produce the outputs, then there is a $(T-1)$-round algorithm for $\re(\Pi)$ where nodes of degree $2$ produce the outputs.

\paragraph{Computable vs non-computable algorithms.}
We previously said that a $T$-round algorithm $A$ for $\Pi$ is a function $f$ that takes as input the $T$-radius neighborhood of a white node (which has degree $\Delta$) and produces an output for each of its incident edges. Note that we did not require $f$ to be a \emph{computable} function. In fact, we allow it to be non-computable (though we require it to be measurable), and our lower bounds hold even in the setting where we allow non-computable algorithms (and hence they are stronger lower bounds). We will \emph{exploit} non-computability when defining an algorithm for $\re(\Pi)$: if we allow an unbounded number of random bits for each node, the algorithm that we define may not be computable. However, it does not matter, since the final result is a lower bound that holds \emph{even for non-computable algorithms}, and hence for computable ones as well.

\section{Automatic Self Reduction}\label{sec:asr}
In this section, we prove the following theorem.
\begin{theorem}\label{thm:newlifting}
Let $\Pi$ be a $b$-dominant $2$-separated problem. 
Let $\varepsilon_0$ be a lower bound on the expected error of any $0$-round algorithm for $\Pi$ that has $0$ white error.
Let $\mathcal{G}$ be a family of $(\Delta,\delta)$-biregular graphs of girth at least $c \log_\Delta n$, for some constant $c >0$. W.l.o.g., assume $\Delta \ge \delta$.
Then, any algorithm for $\Pi$ on $\mathcal{G}$ with $0$ white error and at most $\varepsilon$ expected error requires $c'\left(\min\left\{\log_\Delta n, \log_{2+b} \varepsilon_0/\varepsilon\right\}\right)$ rounds in the randomized LOCAL model, where $c'>0$ depends only on $c$.
\end{theorem}
We will obtain \Cref{thm:newlifting} by a recursive application of the following statement.
\begin{lemma}\label{lem:onestep}
    Let $\Pi = (\Sigma,\Delta,\delta,\constr{C}_W,\constr{C}_B)$ be a $2$-separated $b$-dominant problem. Assume $\Pi$ can be solved in $T\ge 1$ rounds with expected error at most $\varepsilon$ and white error $0$ in graphs with girth at least $2T+3$. Then, $\Pi' = \re(\Pi)$ can be solved in $T-1$ rounds with expected error at most $16(1+b) \varepsilon$ and black error $0$.
\end{lemma}
We first show that \Cref{lem:onestep} implies \Cref{thm:newlifting}. Then, we prove \Cref{lem:onestep}.

\begin{proof}[Proof of \Cref{thm:newlifting}]
    Let $G$ be any $(\Delta,\delta)$-biregular $2$-colored graph of girth $\Omega(\log_\Delta n)$. Assume there exists a $T$-round algorithm for $\Pi$ with expected error $\varepsilon$, white error $0$, and $T$ sufficiently small compared to the girth. 
    Let $S=2\lceil T/2\rceil$. Then, by applying \Cref{lem:onestep} we obtain that there exists a $(S-1)$-round algorithm for $\re(\Pi)$ with expected error at most $16(1+b) \varepsilon$ and black error $0$. By repeating the same argument for $S$ times, we obtain that there exists a $0$-round algorithm for $\re^S(\Pi)$ with expected error at most $\bigl(16(1+b)\bigr)^S \varepsilon$ and white error $0$. By applying \Cref{obs:2sep-is-fixedpoint}, we get that $\re\bigl(\re(\Pi)\bigr) = \Pi$, and hence there exists a $0$-round algorithm for $\Pi$ with expected error at most $\bigl(16(1+b)\bigr)^{S} \varepsilon$. By assumption, we know that $\bigl(16(1+b)\bigr)^{S} \varepsilon \ge \varepsilon_0$. Hence, $S \ge \log_{16(1+b)} \frac{\varepsilon_0}{\varepsilon} \in \Omega\left(\log_{2+b}\frac{\varepsilon_0}{\varepsilon}\right)$, which implies the same asymptotic bound for $T$, and the claim follows.
\end{proof}

The rest of the section is devoted to proving \Cref{lem:onestep}.
\subsection{The Faster Algorithm}
We define a $(T-1)$-round algorithm $\mathcal{A}'$ for solving $\re(\Pi)$ as a function of a given $T$-round algorithm $\mathcal{A}$ for $\Pi$. Let $v$ be a black node. We define the output of $v$ solely as a function of its $(T-1)$-radius neighborhood $B_{T-1}(v)$.
Node $v$, for each neighboring white node $u$, considers \emph{all possible extensions} of $B_{T-1}(v) \cap B_T(u)$ into $T$-radius neighborhoods of $u$ (that is, all possible random bits assignments for the nodes that are at distance $T$ from $v$ and $T-1$ from $u$, and all the nodes that are at distance $T+1$ from $v$ and $T$ from $u$). For each such extension, it computes the output of $\mathcal{A}$ on $u$. For each possible output $\ell$, we obtain a probability $p_\ell$ that $u$ outputs $\ell$ on the edge $\{u,v\}$. Observe that, while we provided an algorithmic description, $p_\ell$ is just $\pr{\text{Node $u$ outputs $\ell$ towards $v$} \mid B_{T-1}(v)}$. 
Let $p_{i,\ell}$ be the value obtained for label $\ell$ on the $i$th incident edge of $v$ (that we simply call \emph{port $i$}).

We define the \emph{discrepancy} of a label $\ell$ on port $i$ as $\discr(i,\ell) = 1 - p_{i,\ell}$, and we define the discrepancy of a configuration $C = (\ell_1,\ldots,\ell_\delta)$ as $\discr(C) = \sum_{1\le i \le \delta}  \discr(i,\ell_i)$. 
Observe that the discrepancy of a configuration $C$ is the expected number of ports in which $\mathcal{A}$ outputs a label that is not the one specified by $C$.
Note that now the order of the labels in the configuration matters (the discrepancy depends on which port gets which label).
In this section, we will need to compute the distance between two configurations, in a setting where the order matters. Hence, we extend the definition of distance to this setting. For two \emph{tuples} $C_1$ and $C_2$, we define their distance $d(C_1,C_2)$ as the number of positions in which they differ. Moreover, unless otherwise specified, each constraint will now be seen as a set of tuples (e.g., if $\constr{C}_W = \{\{2,1,1\}\}$, we now consider it as $\constr{C}_W = \{(2,1,1),(1,2,1),(1,1,2)\}$).

We define $\mathcal{A}'$ as the algorithm that, among all configurations in $\constr{C}_B$, outputs the configuration $C$ (seen now as a tuple of size $\delta$ and not just as a multiset, because the order matters) that minimizes the discrepancy $\discr(C)$.
Recall that, by \Cref{obs:2sep-is-fixedpoint}, up to renaming of labels, $\re(\Pi) = (\Sigma,\delta,\Delta,\constr{C}_B,\constr{C}_W)$. Hence, the defined algorithm $\mathcal{A}'$ never produces an invalid configuration on the black nodes, and we will later upper bound its white error as a function of the black error of $\mathcal{A}$.

As a side note, the standard black-box lifting theorem of round elimination defines $\mathcal{A}'$ very differently: it first computes $p_{i,\ell}$ for each $i$ and $\ell$ as in our case, but then it takes some fixed threshold $f$ and on each port $i$ it outputs the set containing all the labels $\ell$ satisfying $p_{i,\ell} \ge f$. The threshold $f$ is then chosen to minimize the sum of the failure probabilities of the black and the white nodes. Note that such an algorithm could fail on both black and white nodes.

\subsection{A Simpler Statement}
Before proving \Cref{lem:onestep}, we provide a simpler statement that has a simpler proof. 
In some sense, we now give the best bound that we can prove if we only exploit $2$-separation (and we do not use dominance).
We will bound the error of $\mathcal{A}'$ in terms of its discrepancy. We prove a bound on the discrepancy of each node $v$, as a function of the error produced by $\mathcal{A}$ on $v$. Later, we will improve the bound and make it depend on $b$, and we will obtain \Cref{lem:onestep} by essentially going through all nodes and all possible random bits assignments.

Let $X$ be the random variable denoting the output of $\mathcal{A}$ on the edges incident to $v$, conditioned on a fixed $B_{T-1}(v)$. 
Note that $X$ is a random variable that depends on the random bits seen by the neighbors of $v$ (in $T$ rounds) and not seen by $v$ itself (in $T-1$ rounds).
Let $\varepsilon_v = \expect{d(X,\constr{C}_B) \mid B_{T-1}(v)}$. In other words, $\varepsilon_v$ is the expected error that $\mathcal{A}$ produces on $v$, conditioned on a fixed $(T-1)$-neighborhood of $v$.

Let $C_v$ be the output of $\mathcal{A}'$ on $v$. Let $d_v = \discr(C_v)$ be its discrepancy. 
\begin{lemma}\label{lem:onestep-singlenode-simplified}
    Let $v$ be a black node. Then, $d_v \le 2 \delta \varepsilon_v$.
\end{lemma}
\begin{proof}
    For a configuration $C = (\ell_1,\ldots,\ell_\delta)$, let $\pr{C}$ be the probability that the original algorithm $\mathcal{A}$ outputs $C$ on the edges incident to $v$, conditioned on the fixed $(T-1)$-neighborhood of $v$, $B_{T-1}(v)$. By the assumption on $T$ and the girth, we have that for each pair of neighbors $u_1$ and $u_2$ of $v$, it holds that $\bigl(B_T(u_1) \setminus B_{T-1}(v)\bigr) \cap \bigl(B_T(u_2) \setminus B_{T-1}(v)\bigr) = \emptyset$.
    Hence, by independence, $\pr{C} = \prod_{1\le i \le \delta}p_{i,\ell_i}$.

    In order to provide an upper bound on $d_v$, we consider the following quantity:
    \[
        \sum_{C \in \constr{C}_B} \pr{C} \discr(C).
    \]
    That is, for each valid black configuration, we sum its probability, weighted by its discrepancy.
    Let $q_v$ be the probability that $\mathcal{A}$ produces an invalid configuration on $v$, conditioned on $B_{T-1}(v)$.
    Since $C_v$ minimizes the discrepancy, we obtain the following lower bound:
        \[
        \sum_{C \in \constr{C}_B} \pr{C} \discr(C) \ge \sum_{C \in \constr{C}_B} \pr{C} \discr(C_v) = d_v \sum_{C \in \constr{C}_B} \pr{C} = (1-q_v) d_v.
    \]
    In order to provide an upper bound, we first define the \emph{invalid neighborhood of a valid configuration $C$ w.r.t.\ position $i$}, as follows. Let $C = (\ell_1,\ldots,\ell_\delta) \in \constr{C}_B$ be a valid configuration. $N_i(C)$ is the set of configurations that in all positions $j \neq i$ have label $\ell_j$, and in position $i$ have a label different from $\ell_i$. By the $2$-separated property, all configurations in $N_i(C)$ are invalid, because they differ from $C$ in exactly $1$ position.

    Observe that, for any configuration $C = (\ell_1,\ldots,\ell_\delta) \in \constr{C}_B$, and all $1 \le i \le \delta$, the following holds:
    \[
    \pr{C} \left(1 - p_{i,\ell_i}\right) \le  p_{1,\ell_1} \cdot \ldots \cdot p_{i-1,\ell_{i-1}} \cdot \left(1 - p_{i,\ell_i}\right) \cdot p_{i+1,\ell_{i+1}} \cdot \ldots \cdot p_{\delta,\ell_{\delta}}  = \sum_{Z \in N_i(C)} \pr{Z}.
    \]

    Hence, we obtain the following upper bound:
    \[
     \sum_{C \in \constr{C}_B} \pr{C} \discr(C) =  \sum_{C = (\ell_1,\ldots,\ell_\delta) \in \constr{C}_B}  \sum_{1\le i \le \delta} \pr{C} \left(1 - p_{i,\ell_i}\right) \le \sum_{C \in \constr{C}_B} \sum_{1 \le i \le \delta} \sum_{Z \in N_i(C)} \pr{Z}.
    \]
    We claim that for each possible $Z$, we sum $\pr{Z}$ at most $\delta$ times. In fact, for each $Z$ for which we sum $\pr{Z}$, we must obtain $Z$ by picking some $C \in \constr{C}_B$, then we must pick a position $1 \le i \le \delta$, then we must change the label of $C$ in position $i$ and obtain $Z$. We claim that there are at most $\delta$ choices $(C,i)$ that give $Z$: first, for each $C$, there cannot be two values $i,j$ such that $Z \in N_i(C)$ and $Z \in N_j(C)$, because $Z$ differs from $C$ in exactly one position. Then, for each position $i$, there exists at most one $C$ such that $Z \in N_i(C)$, because otherwise there would be two configurations $C_1,C_2 \in \constr{C}_B$ that, by changing them in a single place, we obtain $Z$, which would imply that $d(C_1,C_2) = 1$. By the $2$-separated property, this is not possible. Since there are $\delta$ possible positions, we obtain that, for each $Z$, we sum $\pr{Z}$ for at most $\delta$ times.

    Moreover, each $Z$ for which we sum $\pr{Z}$ is a configuration not in $\constr{C}_B$, and we know that $\sum_{Z \notin \constr{C}_B} \pr{Z} = q_v$. Hence, 
    \[
    \sum_{C \in \constr{C}_B} \sum_{1 \le i \le \delta} \sum_{Z \in N_i(C)} \pr{Z} \le \delta q_v.
    \] 
    Summarizing, we obtained the following,
    \[
        (1-q_v)d_v \le \sum_{C \in \constr{C}_B} \pr{C} \discr(C)  \le \delta q_v.
    \]
    Also, trivially, $d_v \le \delta$.
    If $q_v < 1$,
    we get the following: \[
    d_v \le \delta q_v / (1-q_v),
    \]
    and hence:
    \[
    d_v \le \min\left\{\delta,  \delta q_v / (1-q_v)\right\} \le 2\delta q_v.
    \]
    Note that, whenever the original algorithm produces an error on a node, in order to obtain a correct configuration on the node it is required to change at least one label. Hence, $q_v \le \varepsilon_v$, and the claim follows. Note that the conclusion trivially holds also if $q_v = 1$.
\end{proof}

\subsection{An Improved Bound that Depends on \texorpdfstring{$b$}{b}}
We devote this subsection to proving the following lemma.
\begin{lemma}\label{lem:onestep-singlenode}
    Let $v$ be a black node. Then, $d_v \le 16 (b+1) \varepsilon_v$.
\end{lemma}
We define the set of \emph{dominant labels} as $\dom = \{ \ell \mid \exists C \in \constr{C}_B \text{ s.t.\ } m_C(\ell) \ge \delta - b\}$, that is, all labels for which there exists at least one valid configuration containing it for at least $\delta - b$ times.

The proof is structured as follows. We first prove that the output of $\mathcal{A}$ does not deviate too much from $\{\ell,\ldots,\ell\}$ for some fixed $\ell$. Then, we consider the simple case in which $\delta$ is small compared to $b$ (and we simply conclude the proof by applying \Cref{lem:onestep-singlenode-simplified}), and then, for the case in which $\delta$ is large compared to $b$, we consider separately the cases $\ell \notin \dom$ and $\ell \in \dom$.

\paragraph{The output of $\mathcal{A}$ is not too far from $\{\ell,\ldots,\ell\}$ for some $\ell$.}
Let $\discr_\ell = \discr\bigl((\ell,\ldots,\ell)\bigr)$, that is, the discrepancy of the (possibly invalid) configuration using only label $\ell$. We prove that there exists some label $\ell^*$ satisfying $\discr_{\ell^*} \le 2(\varepsilon_v + b) + 1$.
    In other words, if we sum the probabilities of not obtaining $\ell^*$, we get at most $2(\varepsilon_v + b) + 1$.
\begin{lemma}
There exists a label $\ell^*$ satisfying $\discr_{\ell^*} \le 2(\varepsilon_v + b) + 1$.
\end{lemma}
\begin{proof}
    Recall that $X$ is the random output of $\mathcal{A}$ conditioned on $B_{T-1}(v)$. Let $X_i$ be the random variable denoting the output of $\mathcal{A}$ on port $i$ of $v$, conditioned on $B_{T-1}(v)$.

    Let $d_\dom(C) = \min_{\ell \in \dom} d\bigl(C,(\ell,\ldots,\ell)\bigr)$, that is, the minimum distance from $C$ to a configuration using a single dominant label.

    We start asking, for how many pairs $(i,j)$ where $i \neq j$, algorithm $\mathcal{A}$ produces an output where $X_i \neq X_j$, in expectation. The following holds:
    \[
    \expect*{\sum_{i\neq j}\mathbf{1}_{\{X_i \neq X_j\}}} \le \expect*{ \delta(\delta-1) - \bigl(\delta - d_{\dom}(X)\bigr)\bigl(\delta - d_{\dom}(X) - 1\bigr) } \le \expect{2 \delta d_{\dom}(X)},
    \]
    where the first inequality provides an upper bound as follows. Let $\ell$ be the label from $\dom$ that appears in (the realization of) $X$ the most. It appears $\delta - d_{\dom}(X)$ times. Then, the number of different pairs is at most the number of pairs different from $(\ell,\ell)$.

    Observe that $d_\dom(C) \le d(C,\constr{C}_B) + b$, because $C$ is within distance $d(C,\constr{C}_B)$ from a valid configuration, and any valid configuration contains some dominant label at least $\delta - b$ times by the $b$-dominant property.

    Similarly, $\expect{d_\dom(X)} \le \varepsilon_v + b$, because $X$ has expected error $\varepsilon_v$, and hence expected distance $\varepsilon_v$ from a valid configuration.
    We thus get the following:
    \[
    \expect*{\sum_{1 \le i \le \delta} \sum_{\substack{1 \le j \le \delta,\\j \neq i}}\mathbf{1}_{\{X_i \neq X_j\}}} \le 2 \delta (\varepsilon_v + b).
    \]
    This means that there exists some $i^*$ such that:
    \[
    \expect*{\sum_{\substack{1 \le j \le \delta,\\j\neq i^*}}\mathbf{1}_{\{X_{i^*} \neq X_j\}}} \le 2 (\varepsilon_v + b).
    \]
    This implies the following:
    \[
      \expect*{\sum_{\substack{1 \le j \le \delta,\\j \neq i^* }}\mathbf{1}_{\{X_{i^*} \neq X_j\}}} = \sum_{\substack{1 \le j \le \delta,\\j\neq i^* }}\pr{X_{i^*} \neq X_j} = \sum_{\ell \in \Sigma} p_{i^*,\ell} \sum_{\substack{1 \le j \le \delta\\j \neq i^*}} (1 - p_{j,\ell}) \le 2 (\varepsilon_v + b).
    \]
    This implies that there exists some $\ell^*$ satisfying:
    \[
    \sum_{\substack{1 \le j \le \delta\\j \neq i^*}} \left(1 - p_{j,\ell^*}\right) \le 2 (\varepsilon_v + b).
    \]
    Since $1-p_{i^*,\ell^*} \le 1$, we thus get:
    \[
    \discr_{\ell^*} = \sum_{\substack{1 \le j \le \delta}} \left(1 - p_{j,\ell^*}\right) \le 2 (\varepsilon_v + b) + 1.
    \]
\end{proof}

    \paragraph{The case $\delta < 8(b+1)$.}
    We start by considering a simple case, that is, $\delta < 8(b+1)$. In this case, by applying \Cref{lem:onestep-singlenode-simplified}, we get $d_v \le 2 \delta \varepsilon_v \le 16(b+1) \varepsilon_v$, concluding the proof.
    Hence, in the following, assume $\delta \ge 8(b+1)$.

    \paragraph{The case $\ell^* \notin \dom$.}
    Assume $\ell^* \notin \dom$. 
    On a high level, we now know that $\mathcal{A}$ produces an output that, in expectation, is not far from $\{\ell^*,\ldots,\ell^*\}$ and $\ell^*$ is not a dominant label, but any valid configuration contains some $\ell \neq \ell^*$ for at least $\delta - b$ times. This implies that $\mathcal{A}$ produces a large error as a function of $\delta$, and in order for this to not be a contradiction we get that $\delta$ is small (i.e., linear in $\varepsilon_v + b +1$).

    More formally, observe that, for any $\ell \in \dom$, $d\bigl(\{\ell,\ldots,\ell\},\{\ell^*,\ldots,\ell^*\}\bigr) = \delta$. Hence, for any configuration $C$, $d_{\dom}(C) + d\bigl(C,\{\ell^*,\ldots,\ell^*\}\bigr) \ge \delta$.
    This implies the following:
    \[
    \delta \le \expect{d_{\dom}(X)} + \expect*{d\bigl(X,\{\ell^*,\ldots,\ell^*\}\bigr)} \le \varepsilon_v + b + \discr_{\ell^*} \le \varepsilon_v + b + 2(\varepsilon_v + b) + 1 = 3\varepsilon_v + 3b +1.
    \]
    We now conclude the proof with a simple case analysis.
    Suppose $\varepsilon_v \ge 1/4$. Then,
    \[
    d_v \le \delta \le 3\varepsilon_v + 3b +1 \le  3\varepsilon_v + 3b +4\varepsilon_v = 7\varepsilon_v + 3b \le 12(b+1)\varepsilon_v,
    \]
    concluding the proof. Hence, assume $\varepsilon_v < 1/4$. This case does not apply, since $8(b+1) \le \delta \le 3\varepsilon_v + 3b + 1 < 3/4 + 3b +1$ gives a contradiction.

    \paragraph{The case $\ell^* \in \dom$.}
    If $\varepsilon_v \ge 1/4$, then the discrepancy $d_v$ of $\mathcal{A}'$ is at most the discrepancy of any configuration $C \in \constr{C}_B$ satisfying $m_C(\ell^*) \ge \delta - b$, which exists by the assumption  $\ell^* \in \dom$ and the definition of $\dom$.
    This discrepancy is at most the discrepancy of $\{\ell^*,\ldots,\ell^*\}$, which is $\discr_{\ell^*}$, plus the distance from it to $C$.
    We thus get the following,
    \[
    d_v \le \discr_{\ell^*} + b \le 2(\varepsilon_v + b) + 1 + b = 2\varepsilon_v + 3b + 1 \le 2\varepsilon_v + 3b + 4\varepsilon_v \le 6\varepsilon_v + 3b \le 12(b+1)\varepsilon_v,
    \]
    which concludes the proof. Hence, assume $\varepsilon_v < 1/4$.
    Observe that the following holds:
    \[
    \discr_{\ell^*} \le 2(\varepsilon_v + b) + 1 < 2\left(\frac{1}{4} + b\right) + 1 = 2b + \frac{3}{2}.
    \]

    \subparagraph{Restricting to valid configurations dominated by $\ell^*$.}
    Informally, we now restrict to the valid configurations that contain $\ell^*$ for at least $\delta - b$ times, and, in such a setting, we prove an analogue of a statement proved in \Cref{lem:onestep-singlenode-simplified}.
    Let $\constr{C}_{\ell^*} = \{C \in \mathcal{C}_B \mid m_C(\ell^*) \ge \delta - b\}$, that is, $\constr{C}_{\ell^*}$ is the set of valid configurations that have $\ell^*$ as dominant label. Let $q_{\ell^*} = \pr{X \notin \constr{C}_{\ell^*}}$. Let $d_{\ell^*} = \min_{C \in \constr{C}_{\ell^*} } \expect{ d(C,X) }$. In other words, consider the setting in which only the configurations in $\constr{C}_{\ell^*}$ are valid. In such a setting, $q_{\ell^*}$ is the probability that $\mathcal{A}$ fails, and $d_{\ell^*}$ is the minimum discrepancy of $\mathcal{A}'$ if it is restricted to output configurations from  $\constr{C}_{\ell^*}$.

    \begin{lemma}
    $d_{\ell^*} \le \left(\discr_{\ell^*} + b +1\right)q_{\ell^*}/\left(1-q_{\ell^*}\right).$
    \end{lemma}
    \begin{proof}
    We consider the following quantity:
    \[
    \sum_{C \in \constr{C}_{\ell^*}} \pr{C} \discr(C).
    \]
    Similarly as in the proof of \Cref{lem:onestep-singlenode-simplified}, we lower bound this quantity as follows:
    \[
    \sum_{C \in \constr{C}_{\ell^*}} \pr{C} \discr(C) \ge d_{\ell^*}\sum_{C \in \constr{C}_{\ell^*}} \pr{C} = \left(1 - q_{\ell^*}\right)d_{\ell^*}.
    \]

    Similarly as in the proof of \Cref{lem:onestep-singlenode-simplified}, for a configuration $C = (\ell_1,\ldots,\ell_\delta) \in \constr{C}_{\ell^*}$, we define  $N_i(C)$ as the set of configurations that in all positions $j \neq i$ have label $\ell_j$, and in position $i$ have a label different from $\ell_i$. By the $2$-separated property, all configurations in $N_i(C)$ are not in $\constr{C}_{\ell^*}$, because they differ from $C$ in exactly $1$ position, and $\constr{C}_{\ell^*} \subseteq \constr{C}_B$.
    Observe that, for any configuration $C = (\ell_1,\ldots,\ell_\delta) \in \constr{C}_{\ell^*}$, and all $1 \le i \le \delta$, the following holds:
    \[
    \pr{C} \left(1 - p_{i,\ell_i}\right) = p_{i,\ell_i}\cdot  \left(p_{1,\ell_1} \cdot \ldots \cdot p_{i-1,\ell_{i-1}} \cdot \left(1 - p_{i,\ell_i}\right) \cdot p_{i+1,\ell_{i+1}} \cdot \ldots \cdot p_{\delta,\ell_{\delta}}\right)  = p_{i,\ell_i}\sum_{Z \in N_i(C)} \pr{Z},
    \]
    which is similar to the bound used in \Cref{lem:onestep-singlenode-simplified}, except that there we upper bounded $p_{i,\ell_i}$ by $1$. In \Cref{lem:onestep-singlenode-simplified}, we continued by upper bounding (by $\delta$) how many times we sum $\pr{Z}$ for each $Z$. We now proceed similarly, except that now we perform a count that is weighted by $p_{i,\ell_i}$.

    Let $C^{i\rightarrow \ell}$ be the configuration obtained by replacing the $i$th label of $C$ with $\ell$.
    We obtain the following equality: 
    \begin{align*}
        \sum_{C \in \constr{C}_{\ell^*}} \pr{C} \discr(C) &=  \sum_{C = (\ell_1,\ldots,\ell_\delta) \in \constr{C}_{\ell^*}}  \sum_{1\le i \le \delta} \pr{C} \left(1 - p_{i,\ell_i}\right)
        \\&= \sum_{C = (\ell_1,\ldots,\ell_\delta) \in \constr{C}_{\ell^*}} \sum_{1 \le i \le \delta} p_{i,\ell_i} \sum_{Z \in N_i(C)} \pr{Z} 
        \\&=\sum_{\substack{(C,i,Z) ~: \\ C = (\ell_1,\ldots,\ell_\delta) \in \constr{C}_{\ell^*} \\\land ~ 1 \le i \le \delta\\ \land ~ Z \in N_i(C)}} p_{i,\ell_i} \pr{Z}
        \\&=\sum_{Z  \notin \constr{C}_{\ell^*}} \pr{Z} \sum_{\substack{(C,i) ~:\\C = (\ell_1,\ldots,\ell_\delta) \in \constr{C}_{\ell^*}\\ \land~ 1 \le i \le \delta\\\land~ Z \in N_i(C) }} p_{i,\ell_i} 
        \\&=\sum_{Z \notin \constr{C}_{\ell^*}} \pr{Z} \cdot W(Z), \text{ where }W(Z) = \sum_{(i,\ell) ~:~ Z^{i\rightarrow \ell} \in \constr{C}_{\ell^*}}p_{i,\ell},
    \end{align*}
    where in the fourth equality we used the $2$-separated property to infer that if $Z \in N_i(C)$ then $Z \notin  \constr{C}_{\ell^*}$. We call $W(Z)$ the \emph{weight} of $Z$.
    In the proof of \Cref{lem:onestep-singlenode-simplified}, we essentially upper bounded the weight of $Z$ with $\delta$, but we now provide a better bound.

    Note that $\ell$ cannot be equal to the $i$th label of $Z$, since $Z \notin \constr{C}_{\ell^*}$ but $Z^{i\rightarrow \ell} \in \constr{C}_{\ell^*}$.
    Hence, let $Z = (\ell_1,\ldots,\ell_\delta)$. We get that \[
    W(Z) \le \sum_{1 \le i \le \delta} \left(1 - p_{i,\ell_i}\right) = \discr(Z),
    \]
    that is, the weight of $Z$ is upper bounded by its discrepancy, which can be upper bounded by the discrepancy of $(\ell^*,\ldots,\ell^*)$, plus the distance of $Z$ from $(\ell^*,\ldots,\ell^*)$.
    Let $k(Z) = |\{i \mid \ell_i \neq \ell^*\}|$.
    We thus get the following:
    \begin{align*}
    W(Z) &\le \discr_{\ell^*} + k(Z).
    \end{align*}
    Now observe that, if $W(Z) > 0$, then $Z \in N_i(C)$ for some $i$ and some $C \in \constr{C}_{\ell^*}$. Hence, by changing $1$ label of $Z$ we obtain a configuration in $ \constr{C}_{\ell^*}$, and by changing at most $b$ additional labels we reach $\{\ell^*,\ldots,\ell^*\}$. Hence, if $W(Z) > 0$, we get that $k(Z) \le b+1$.
    We obtain the following:
    \[
    \sum_{C \in \constr{C}_{\ell^*}} \pr{C} \discr(C) = \sum_{Z \notin \constr{C}_{\ell^*}} \pr{Z} \cdot W(Z) \le \sum_{Z \notin \constr{C}_{\ell^*}} \pr{Z} \left(\discr_{\ell^*} + b + 1\right) = \left(\discr_{\ell^*} + b + 1\right)q_{\ell^*}.
    \]
    Hence,
    \[
    \left(1 - q_{\ell^*}\right)d_{\ell^*} \le \sum_{C \in \constr{C}_{\ell^*}} \pr{C} \discr(C) \le \left(\discr_{\ell^*} + b +1\right)q_{\ell^*},
    \]
    as required.
    \end{proof}

    \subparagraph{Considering again all the valid configurations.}
    We now bound $q_{\ell^*}$ as a function of $\varepsilon_v$. Observe that, for each configuration $C = (\ell_1,\ldots,\ell_\delta) \in \constr{C}_B \setminus \constr{C}_{\ell^*}$, it holds that \[
    k(C) = |\{i \mid \ell_i \neq \ell^*\}| \ge \delta - b,
    \]
    by the $b$-dominant property and the fact that $C$ is valid but not in $\constr{C}_{\ell^*}$.
    Hence, any configuration $Z$ such that $b < k(Z) < \delta - b$ is invalid.
    Hence,\[
     \pr{b < k(X) < \delta - b} \le \varepsilon_v.
    \]
    
    Recall that $\delta \ge 8(b+1)$ and that $\discr_{\ell^*} < 2b + \frac{3}{2}$.
    Assume that \(\pr{k(X) \ge b+1} > 0\).
    We now prove that $\expect{k(X) \mid k(X) \ge b+1 } \le \discr_{\ell^*} + b+1$.

    Define \(Y_i\) to be the indicator variable of the event that \(X_i \neq \ell^*\). 
    The variables \(Y_1, \ldots, Y_\delta\) are independent and \(k(X) = \sum_{i = 1}^\delta Y_i\).
    Also, \(\expect{k(X)} = \sum_{i = 1}^\delta \pr{Y_i = 1} = \discr_{\ell^*}\).
    For each \(i \in \{1, \ldots, \delta\}\), define \(k_{-i}(X) = \sum_{j \neq i} Y_j\).
    For an event \(A\), let \(\mathds{1}(A)\) be the indicator random variable for event \(A\).
    By independence, 
    \[
        \expect{k(X) \mathds{1}(\{k(X) \ge b+1\})}= \sum_{i = 1}^\delta \pr{Y_i = 1} \pr{k_{-i}(X) \ge b}.
    \]
    Now, \(\pr{k_{-i}(X) \ge b} = \pr{k_{-i}(X) \ge b+1} + \pr{k_{-i}(X) = b}\).
    Note that \(\pr{k_{-i}(X) \ge b+1} \le \pr{k(X) \ge b+1}\).
    Hence,
    \begin{align*}
        \expect{k(X) \mathds{1}(\{k(X) \ge b+1\})} & \le \sum_{i = 1}^\delta \pr{Y_i = 1} \pr{k(X) \ge b+1} + \sum_{i = 1}^\delta \pr{Y_i = 1} \pr{k_{-i}(X) = b}.
    \end{align*}
    Note that
    \begin{align*}
        \sum_{i = 1}^\delta \pr{Y_i = 1} \pr{k_{-i}(X) = b} & = \sum_{i = 1}^\delta \pr{Y_i = 1,k(X) = b+1} \\
        & = \expect{\sum_{i = 1}^\delta Y_i \mathds{1}(\{k(X) = b+1\})} \\
        & = (b+1) \pr{k(X) = b+1} \\
        & \le (b+1) \pr{k(X) \ge b+1}.
    \end{align*}
    Therefore,
    \begin{align*}
        \expect{k(X) \mathds{1}(\{k(X) \ge b+1\})} & \le \left(\sum_{i = 1}^\delta \pr{Y_i = 1}  + b+1 \right) \pr{k(X) \ge b+1} \\
        & = \left(\discr_{\ell^*} + b+1\right)\pr{k(X) \ge b+1}.
    \end{align*}
    Since by definition
    \[
        \expect{k(X) \mid k(X) \ge b+1 } = \frac{\expect{k(X)\mathds{1}(\{k(X) \ge b+1\})}}{\pr{k(X) \ge b+1}},
    \]
    dividing both sides of the previous inequality by \(\pr{k(X) \ge b+1}\) gives the desired claim.

    Now assume $\pr{k(X)\ge b+1}=0$. 
    Then $\pr{k(X)\ge\delta-b}=0\le\varepsilon_v$, so the desired bound holds trivially. 
    Otherwise, by Markov's inequality under the conditional distribution given \(k(X) \ge b+1\), we get the following:
    \begin{align*}
    \frac{\pr{k(X) \ge \delta - b}}{\pr{k(X) \ge b+1}}  &= \frac{\pr{k(X) \ge \delta - b \land k(X) \ge b+1}}{\pr{k(X) \ge b+1}} =   \pr{k(X) \ge \delta - b \mid k(X) \ge b+1 } \\&\le \frac{\discr_{\ell^*} + b + 1}{\delta-b} \le 
    \frac{3b + \frac{5}{2}}{\delta-b} \le \frac{3b + \frac{5}{2}}{7b+8}
    \le 
    \frac{1}{2}.
    \end{align*}
    Hence, 
    \[
    2\pr{k(X) \ge \delta - b} \le \pr{k(X) \ge b+1} =  \pr{b < k(X) < \delta -b} + \pr{ k(X) \ge \delta -b}.
    \]
    This implies \[
    \pr{k(X) \ge \delta - b} \le \pr{b < k(X) < \delta -b} \le \varepsilon_v.
    \]
    We get that \[
    q_{\ell^*} \le \pr{X \notin \constr{C}_B} + \pr{ k(X) \ge \delta-b } \le 2 \varepsilon_v < \frac{1}{2}.
    \]
    Finally,
    \begin{align*}
        d_v &\le d_{\ell^*} \le \frac{\left(\discr_{\ell^*} + b +1\right)q_{\ell^*}}{1-q_{\ell^*}} \le 2\left(\discr_{\ell^*} + b +1\right)q_{\ell^*} \\&\le 4\left(\discr_{\ell^*} + b +1\right)\varepsilon_v \le 4\bigl((2b + 2) + b +1\bigr) \varepsilon_v
        = 12(b+1)\varepsilon_v,
    \end{align*}
    concluding the proof.

\subsection{Proof of \texorpdfstring{\Cref{lem:onestep}}{the Main Lemma}.}
The proof is essentially obtained by summing the bound of \Cref{lem:onestep-singlenode} over all nodes, and considering all possible random bits assignments.

In order to prove this statement, we upper bound the error of the algorithm $\mathcal{A}'$ that we previously defined. 
In order to bound the error of $\mathcal{A}'$, we bound its expected deviation from the output of $\mathcal{A}$, that is, in expectation, on how many edges $\mathcal{A}'$ produces an output that is not exactly the same as the output produced by $\mathcal{A}$. Suppose the expected deviation is $d$. Then, we can upper bound the error of  $\mathcal{A}'$ as $d$ because:
    \begin{itemize}
        \item $\mathcal{A}'$ has $0$ black error.
        \item Given an output of $\mathcal{A}'$, in expectation we need to change $d$ labels to obtain the output of $\mathcal{A}$. But $\mathcal{A}$ has $0$ white error, and hence we need to change $d$ labels in expectation to produce a solution with $0$ white error. 
        \item The error of $\mathcal{A}'$ is the sum, over all black nodes and all white nodes, of how many labels we need to change at most to make the node happy. We hence have that the black error is $0$ and the expected white error is upper bounded by $d$.
    \end{itemize}
    Hence, the claim follows if we prove that $d \le 16(b+1) \varepsilon$. 
    Let $B$ be the set of black nodes.
    Observe that $\expect*{\sum_{v \in B} \varepsilon_v} \le \varepsilon$, and $\expect*{\sum_{v \in B} d_v} = d$, where the expectation is taken over all possible random bit assignments to the nodes.
    Hence, by \Cref{lem:onestep-singlenode} and linearity of expectation,
    \[
        d = \expect*{\sum_{v \in B} d_v} \le 16 (b+1) \expect*{\sum_{v \in B} \varepsilon_v} \le 16 (b+1) \varepsilon.
    \]
    Hence, the claim follows.

\section{Applications}\label{sec:applications}
In this section, we provide some applications of \Cref{thm:newlifting}.
We will exploit the following high-girth graph construction, for which we provide an explicit proof for the exact parameter range we need.
\begin{lemma}\label{lem:graph-family-1}
    There exist three positive constants $\alpha$, $\Delta_0$, and $\varepsilon$ such that, for every sufficiently large $n$ and every $\Delta$ satisfying that $n\Delta$ is even and that $\Delta_0\le\Delta\le n^{1/4}$, there exists a $\Delta$-regular graph on $n$ nodes that has girth at least $\varepsilon \log_\Delta n$ and independence number at most $\alpha \frac{n \log \Delta}{\Delta}$.
\end{lemma}
\begin{proof}[Proof sketch]
    The first part of the main theorem of Frieze and Łuczak~\cite{frieze-luczak-1992} implies that, for every sufficiently large \(n\) and every \(\Delta_0\le\Delta\le n^{1/4}\) with \(n\Delta\) even, a uniformly random simple \(\Delta\)-regular graph on \(n\) nodes has independence number at most \(O(n\log\Delta/\Delta)\) with high probability.
    Fix a sufficiently small auxiliary constant $\eta>0$.
    If $\eta\log_\Delta n<4$, the graph is simple and hence has girth at least $3>\frac{\eta}{2}\log_\Delta n$. Thus, this case is already covered by taking $\varepsilon=\eta/2$.
  Otherwise, for \(g=\lfloor\eta\log_\Delta n\rfloor\), the condition \((\Delta-1)^{2g-1}=o(n)\) of McKay, Wormald, and Wysocka~\cite{mckay-wormald-wysocka-2004} is satisfied.
    Lemma~1 of their paper, together with the Poisson means defined in their equation~(1.2), implies that such a graph has at most \(O(n^{4\eta})\) cycles of length less than \(g\), with high probability.
    Hence, there exists a \(\Delta\)-regular graph satisfying both properties simultaneously.

    Following the degree-preserving switching used in the proof of Lemma~2.1 of Alon~\cite{alon-2010-on-constant-time-approximation-of-parameters}, we remove all its short cycles.
    The switchings preserve regularity and can be chosen so that they create no new short cycles.
    Since every switching removes only two edges, after \(O(n^{4\eta})\) switchings the independence number has increased by at most \(O(n^{4\eta})\).
    For a sufficiently small choice of \(\eta\), this is \(o(n\log\Delta/\Delta)\), so the resulting graph has the claimed girth and independence number after adjusting the constants.
    The full argument is given in the appendix (see \Cref{lem:app:high-girth-independent}).
\end{proof}

On a high level, all our results are obtained by applying the following idea, which is a direct application of \Cref{thm:newlifting}.
\begin{observation}\label{obs:local-to-global}
    Let $\Pi^L$ be a problem, and $\Pi^G$ be a problem in the black-white formalism. Suppose that $\Pi^G$ is $2$-separated and $b$-dominant, and suppose that any $0$-round algorithm for $\Pi^G$ that has $0$ white error has expected black error at least $\varepsilon_0$.
    Let $\mathcal{G}$ be a family of $(\Delta,\delta)$-biregular graphs of girth $\Omega(\log_\Delta n)$, and assume, w.l.o.g., that $\Delta \ge \delta$.
    Suppose that, on $\mathcal{G}$, there is a $T$-round algorithm mapping a solution of $\Pi^L$ into a solution of $\Pi^G$, such that the expected black error is at most $\varepsilon$ and the white error is $0$.
    Then, $\Pi^L$, on $\mathcal{G}$, requires $\Omega\left(\min\left\{\log_\Delta n, \log_{2+b} \frac{\varepsilon_0}{\varepsilon}\right\}\right) - T$ rounds.
\end{observation}
In other words, by exploiting \Cref{thm:newlifting}, we can directly prove lower bounds for \emph{local} problems, by exploiting a black-box lower bound that we obtain for \emph{global} problems. Essentially, the lower bound that we obtain is the logarithm of the ratio $\varepsilon_0/\varepsilon$ that we introduce in the mapping. For example, a lower bound for maximal matching immediately derives from the following observations:
\begin{itemize}
    \item Let $\Pi^L$ be the problem of maximal matching.
    \item Let $\Pi^G$ be the problem of perfect matching. It is $2$-separated and $1$-dominant.
    \item By \Cref{lem:graph-family-1}, there are families of graphs of girth $\Omega(\log_\Delta n)$ and independence number $O(n \log \Delta / \Delta)$. If we solve $\Pi^L$ in these graphs, unmatched nodes form an independent set, and hence they are at most $O(n \log \Delta / \Delta)$.
    \item Any $0$-round algorithm for $\Pi^G$ produces $\Omega(n)$ errors in expectation.
    \item Consider the following mapping from  $\Pi^L$ to  $\Pi^G$: unmatched nodes mark an arbitrary incident edge as matched. This produces $O(n \log \Delta / \Delta)$ errors.
    \item We get that $\log \frac{\varepsilon_0}{\varepsilon} \in \Omega(\log \Delta)$. 
\end{itemize}

Hence, in order to prove a lower bound, it is sufficient to find a suitable pair $(\Pi^L,\Pi^G)$ of problems.
We now provide some applications of this idea. 

An easy application of our new black-box lifting theorem is a lower bound for maximal $\frac{1}{k}$-integral matching, which is the problem of computing an assignment from $\{0,\frac{1}{k},\frac{2}{k},\ldots,1\}$ to each edge, such that the values assigned to the edges incident to a node sum to at most $1$, and no edge can be increased. 
\kintegral*
\begin{proof}
   We consider the setting in which $\Delta > k$.
   Let $\Pi^L$ be the  maximal $\frac{1}{k}$-integral matching problem. Let $\Pi^G$ be the \emph{perfect} $\frac{1}{k}$-integral matching problem, that is, the problem of outputting a fraction from $\{0,\frac{1}{k},\frac{2}{k},\ldots,1\}$ on each edge, such that, for each node, the sum of the values assigned to its incident edges is exactly $1$. We call \emph{unsaturated} nodes the nodes for which their edges do not sum to exactly $1$.

   We claim that $\Pi^G$ is $2$-separated and $(k+1)$-dominant.
   Its edge constraint contains all pairs $\{\frac{i}{k}, \frac{i}{k}\}$ for $0 \le i \le k$, which clearly satisfies $2$-separation (and trivially it is $2$-dominant, and $2 \le k+1$).
   Its node constraint contains all multisets of size $\Delta$ where each label is from $\{0,\frac{1}{k},\frac{2}{k},\ldots,1\}$ and the sum is exactly $1$.
   Observe that $k$-dominance is satisfied: each configuration $C$ must satisfy $m_C(0) \ge \Delta - k$.
   Now, suppose for a contradiction that the node constraint is not $2$-separated. This means that there exist two configurations $C_1 = \{\ell_1,\ldots,\ell_\Delta\}$ and $C_2 = \{\ell'_1,\ldots,\ell'_\Delta\}$ that satisfy $d(C_1,C_2) = 1$. This means $\sum_{1 \le i \le \Delta} \ell_i \neq \sum_{1 \le i \le \Delta} \ell'_i$, and hence that at least one of the two configurations does not sum to $1$, a contradiction.

   Consider the family of graphs from \Cref{lem:graph-family-1}. In such graphs, the independence number is $O(n \log \Delta / \Delta)$. In any solution for $\Pi^L$, it must hold that unsaturated nodes form an independent set: for a contradiction, suppose $u$ and $v$ are neighboring unsaturated nodes; we could increase the edge between them by at least $1/k$, contradicting maximality.

   Consider the following mapping from $\Pi^L$ to $\Pi^G$: each unsaturated node $v$ picks an arbitrary node-edge pair labeled $0$ and increases it such that $v$ becomes saturated. We obtain a solution for $\Pi^G$ in which there are no errors on the nodes, and $O(n \log \Delta / \Delta)$ errors on the edges, because there are $O(n \log \Delta / \Delta)$ unsaturated nodes, and each of them modifies a single edge.

   Now consider an arbitrary $0$-round algorithm for $\Pi^G$ with $0$ white error.
   Any $0$-round algorithm must produce an output solely as a function of the random bits it sees on the node.
   On each node, since $\Delta > k$, and since it must output labels that sum to $1$, it must output $0$ on at least one edge, and non-$0$ on at least one edge. This implies that, if we look at an arbitrary edge, we get that it must get different labels with probability at least $1/\Delta$. Since the number of edges is $n \Delta /2$, we get that the expected error is $\Omega(n)$.
    Hence, the claim follows by applying \Cref{obs:local-to-global} (and the $\Omega\left(\sqrt{\log n}\right)$ lower bound is obtained by picking $\Delta = 2^{\sqrt{\log n}}$).
\end{proof}

We now consider the maximal $H$-packing problem, where the goal is to find a maximal collection of vertex-disjoint copies of $H$ in the input graph, where $H$ is a fixed tree. In the following, we assume that the degree of $H$ is strictly less than the degree of $G$ (or, in other words, in order to prove our lower bound, we consider a family of graphs with sufficiently large degree).
\hpacking*
\begin{proof}
    Let $\Pi^L$ be the maximal $H$-packing problem.
    Let $\Pi$ be the perfect $H$-packing problem, that is,  maximal $H$-packing restricted to solutions in which each node is in a copy of $H$.
    We prove that there exists a problem $\Pi^G$ in the black-white formalism that is $2$-separated and $O_H(1)$-dominant, and such that, given a solution to $\Pi$, we can solve $\Pi^G$ in $O_H(1)$ rounds, where $O_H(1)$ hides the dependencies on $H$.

    We define a labeling $\ell : V_H \times E_H \rightarrow \Sigma$ for some finite set $\Sigma$ of labels to be defined later.
    Consider the incidence graph $H' = (V_{H'},E_{H'})$ of $H$, that is, the graph where there is a white node for each node in $V_H$ and a black node for each edge in $E_H$, connected in the natural way. Let $c \in V_{H'}$ be the node with minimum eccentricity in $H'$ (note that this node is unique).
    For each node $v \in V_{H'}$, let $p(v)$ be the edge on the shortest path from $v$ to $c$, and let $p(c)$ be undefined.
    For each node $v \neq c$, let $t(v)$ be the subtree rooted at $v$ and induced by the nodes reachable from $v$ by not using the edge $p(v)$. Let $\mathcal{T} = \{t(v) \mid v \in V_{H'}\setminus \{c\}\}$ be the set of subtrees appearing in $H'$, where two trees are considered equal if they are equal up to isomorphism (where the isomorphism preserves the root and the colors of the nodes). Let $\sigma$ be an arbitrary bijective function from $\mathcal{T}$ to $\{1,\ldots,|\mathcal{T}|\}$. In other words, $\sigma$ assigns a unique integer to each possible \emph{type} of subtree.
    For each edge $e = \{u,v\}$, let $u$ be the node satisfying $p(u) = e$. We label $e$ with $\sigma(t(u))$. Then, we map back this labeling to $H$ in the natural way (and $\Sigma$ is the set of labels that we used).

    We define $\Pi^G$ as follows. The set of labels is $\Sigma \cup \{\bot\}$.
    If $H$ contains an edge labeled $\{\ell,\ell'\}$, we allow the edge configuration $\{\ell,\ell'\}$. Moreover, we allow the edge configuration $\{\bot,\bot\}$.
    If $H$ contains a node labeled $\{\ell_1,\ldots,\ell_d\}$, we allow the node configuration $\{\ell_1,\ldots,\ell_d,\bot,\ldots,\bot\}$ of size $\Delta$.
    \Cref{fig:perfect-h-packing-encoding} gives an example of this encoding.
    \input{figures/perfect-h-packing.tex}
    We claim that $\Pi^G$ is $2$-separated and $O_H(1)$-dominant. 
    
    We first handle the configuration at the center $c$.
    Let $r$ be the radius of $H'$. Since $c$ is its unique center,
    at least $2$ subtrees rooted at children of $c$ have height exactly $r-1$,
    while all subtrees rooted at distance $\ge2$ from $c$
    have height $\le r-2$.
    If $c$ is black, it has exactly two children, and both corresponding
    subtree types occur only at children of $c$. Hence, both labels in
    the edge configuration of $c$ occur in no other edge configuration,
    so its distance from every other edge configuration is $2$.
    Similarly, if $c$ is white, its node configuration contains at least two
    labels (possibly the same twice) that occur in no other node
    configuration.
    In either case, its configuration has distance at least two
    from every other configuration of the same constraint.
    We can thus restrict the remaining argument to configurations
    arising at nodes other than $c$.

    Suppose for a contradiction that $2$-separation does not hold on the edge constraint. Then there exist two allowed configurations $\{\ell,\ell_1\}$ and $\{\ell,\ell_2\}$. This implies that there exist two subtrees $t_1$ and $t_2$ of type $\ell$ such that, if we connect a black node to the root of $t_1$ we obtain a different tree from the one obtained by adding a black node to the root of $t_2$, which is a contradiction, because $t_1 = t_2$.
    Now consider the node configurations. Suppose there exist two configurations $C_1$ and $C_2$ such that $d(C_1,C_2) = 1$. From $C_1$ and $C_2$ we extract the two trees from which they have been constructed, $t_1$ and $t_2$.
    They must have different types, $\ell_1$ and $\ell_2$.
    Consider two cases, either the roots of $t_1$ and $t_2$ have the same degree, or not. In the former case, since $t_1$ and $t_2$ have different types but same degree, this implies that the types of the children of the roots of $t_1$ and $t_2$ differ. Hence, there are at least two differences between $C_1$ and $C_2$, the label of the root and the label of the child.
    In the latter case, $C_1$ and $C_2$ are of the following form. $C_1 = \{\ell_1,\ldots,\ell_d,\bot,\ldots,\bot\}$ and $C_2 = \{\ell'_1,\ldots,\ell'_{d'},\bot,\ldots,\bot\}$ with $d \neq d'$, w.l.o.g.\ $d < d'$. Moreover, it cannot hold that $\{\ell_1,\ldots,\ell_d\} \subset \{\ell'_1,\ldots,\ell'_{d'}\}$, since this would require $\sigma$ to assign the same label to a tree rooted at a black node and to a tree rooted at a white node. Hence, there are at least two differences.

    Consider an arbitrary solution of $\Pi^L$, and let $U$ be the subgraph induced by nodes that are not part of any copy of $H$. 
    We prove that $U$ can be colored with $O_H(1)$ colors.
    Let $h=|V(H)|$. By maximality, $U$ contains no copy of $H$.
  Every graph of minimum degree at least $h-1$ has at least $h$
  vertices and contains every $h$-vertex tree as a subgraph,
  by greedy embedding. Thus, no subgraph of $U$ has minimum
  degree at least $h-1$, so $U$ is $(h-2)$-degenerate and hence
  $(h-1)$-colorable.
    
   Consider the family of graphs from \Cref{lem:graph-family-1}. In such graphs, the independence number is $O(n \log \Delta / \Delta)$. Hence, any subset of nodes that is $O_H(1)$ colorable includes at most a fraction $O(\log \Delta / \Delta)$ of the nodes.

    Consider the following mapping from $\Pi^L$ to $\Pi^G$: 
    nodes that are part of some $H$ spend $O_H(1)$ rounds to compute their labeling in $H$, and output the associated configuration.
    Nodes not in $H$ output the configuration that a leaf of $H$ would output, in an arbitrary way.
    We obtain a solution for $\Pi^G$ in which there are no errors on the nodes, and $O(n \log \Delta / \Delta)$ errors on the edges, because there are $O(n \log \Delta / \Delta)$ nodes that are not part of any $H$, and each of them can create at most one
     edge error (on the edge not getting $\bot$).

    Now consider an arbitrary $0$-round algorithm for $\Pi^G$.
   Any $0$-round algorithm must produce an output solely as a function of the random bits it sees on the node.
   Observe that each valid node configuration contains at least two labels ($\bot$ and something else), and each label is edge compatible only with a single other label.
   Hence, an edge gets an error with probability at least $1/\Delta$, and thus the expected error is $\Omega(n)$.
    Hence, the claim follows by applying \Cref{obs:local-to-global} (and the $\Omega\left(\sqrt{\log n}\right)$ lower bound is obtained by picking $\Delta = 2^{\sqrt{\log n}}$).
\end{proof}

We now consider the bipartite maximal matching problem. While a lower bound of $\Omega(\log \Delta)$ is already known for maximal matching, we now show that it can be extended to bipartite $2$-colored graphs. For this purpose, we first prove the following statement.
\begin{lemma}
    For every sufficiently large even $n$ and every $3\le\Delta\le n^{1/4}$, there exists a bipartite $\Delta$-regular graph on $n$ nodes of girth $\Omega(\log_\Delta n)$ in which every maximal matching leaves at most an $O(\log\Delta/\Delta)$ fraction of the nodes unmatched.
\end{lemma}
The proof of this lemma is deferred to the appendix. See \Cref{lem:app:bipartite-high-girth}.
It is similar to that of \cref{lem:app:high-girth-independent}; however, since we need to focus on bipartite graphs, we need to change the random graph model of interest, and we consider the bipartite configuration model.  

\bmm*
\begin{proof}
    Consider the bipartite perfect matching problem. In the black-white formalism it can be expressed by using the constraints $\constr{C}_W = \constr{C}_B = \{\{\mathrm{M},\mathrm{U},\ldots,\mathrm{U}\}\}$. That is, each edge is either labeled $\mathrm{M}$ (matched) or $\mathrm{U}$ (unmatched), and each node needs to have exactly one $\mathrm{M}$.
    Similarly as in the previous proofs, the error that we obtain when mapping bipartite maximal matching to perfect matching is $O(n \log \Delta / \Delta)$, because there are $O(n \log \Delta / \Delta)$ unmatched nodes, and again, any $0$-round algorithm has $\Omega(n)$ expected error. Thus, the claim follows.
\end{proof}

\section{Optimality of Our Approach}\label{sec:optimality}
We now discuss the optimality of our approach.
Recall that, in order to upper bound the error of the faster algorithm, we provided an upper bound on the discrepancy of its output from the output of the original algorithm.
We first show that, unless we change how we measure the error, $2$-separation is somehow necessary to provide a good upper bound on the discrepancy.
Let $\Pi = (\Sigma,\Delta,\delta,\constr{C}_W,\constr{C}_B)$.
If $\Pi$ is not $2$-separated, a valid black configuration of $\re(\Pi)$ is not necessarily composed of singletons, and in general it consists of nonempty label sets $L_1,\ldots,L_\delta$ such that every choice $\ell_i \in L_i$ gives a configuration $\{\ell_1,\ldots,\ell_\delta\} \in \constr{C}_B$.
Given random labels $X_1,\ldots,X_\delta$, we define the \emph{discrepancy} of this configuration, with $L_i$ assigned to port $i$, as
  \[
      \discr(L_1,\ldots,L_\delta)
      = \sum_{i=1}^{\delta} \pr{X_i \notin L_i}.
  \]
  This is the expected number of ports on which the original output lies outside the chosen label set, and it agrees with our earlier definition when every $L_i$ is a singleton.

  We restrict to the case $\delta=2$. For any pair of sets \(A,B\), we define their unordered product \(A \odot B\) as the set \(\{\{a,b\}: a \in A, b \in B\}\), where the elements are multisets.

  Moreover, for two labels $\ell_1,\ell_2$ we say that $\ell_1$ and $\ell_2$ are \emph{equivalent} if, for any $C \in \constr{C}_B$, if we replace any number of occurrences of $\ell_1$ in $C$ with $\ell_2$, and any number of occurrences of $\ell_2$ in $C$ with $\ell_1$, we obtain a configuration $C' \in  \constr{C}_B$. It is known that, if a problem $\Pi$ uses equivalent labels, then $\Pi$ can be normalized into an equivalent problem that does not use distinct equivalent labels.

  We show that, after this normalization, if the black constraint is not $2$-separated, then some independent label distributions have expected error at most $\varepsilon$, while every valid choice of label sets has discrepancy $\Omega(\sqrt{\varepsilon})$.
  Consequently, the discrepancy cannot always be bounded by a constant multiple of the original expected error.
  In more detail, we prove the following theorem.

\begin{theorem}\label{thm:limit1}
  Let $\Pi = (\Sigma,\Delta, 2,\constr{C}_W,\constr{C}_B)$ be a problem that does not contain distinct equivalent labels and whose black constraint is not $2$-separated.
  Then, for every sufficiently small $\varepsilon > 0$, there exist probability
  distributions $(p_{i,\ell})_{\ell\in\Sigma}$, one for each
  $i\in\{1,2\}$, with the following property.

  Let $X_1,X_2$ be independent random variables with
  $\pr{X_i=\ell}=p_{i,\ell}$. Then
  \[
      \expect*{
          d\bigl(\{X_1,X_2\},\constr{C}_B\bigr)
      }\le\varepsilon,
  \]
  but, for every two non-empty subsets \(L_1, L_2 \subseteq \Sigma\) such that \(L_1 \odot L_2 \subseteq \constr{C}_B\), we have that
  \[
      \discr(L_1,L_2) \ge \sqrt{\varepsilon}.
  \]
  \end{theorem}
\begin{proof}
    Since \(\constr{C}_B\) is not 2-separated, there exist two configurations \(\{x,z\},\{y,z\} \in \constr{C}_B\) such that \(x \neq y\).
    Also, since \(\constr{C}_B\) does not contain distinct equivalent labels, there exists a label \(w \in \Sigma\) such that, w.l.o.g., \(\{x,w\} \in \constr{C}_B\) and \(\{y,w\} \notin \constr{C}_B\).
    Define
    \begin{align*}
        X_1  & = 
        \begin{cases}
            x &  \quad \text{with probability } 1-\sqrt{\varepsilon}; \\
            y & \quad \text{with probability } \sqrt{\varepsilon};
        \end{cases}
        \\
        X_2  & = 
        \begin{cases}
            z &  \quad \text{with probability } 1-\sqrt{\varepsilon}; \\
            w & \quad \text{with probability } \sqrt{\varepsilon}.
        \end{cases}
    \end{align*}
    The configuration \((X_1,X_2)\) is valid if and only if \((X_1,X_2) \neq (y,w)\).
    
    Consider any pair of sets \(L_1, L_2 \subseteq \Sigma\) such that \(L_1 \odot L_2 \subseteq \constr{C}_B\).
    We cannot have that \(y \in L_1\) and \(w \in L_2\), because otherwise \(\{y,w\} \in \constr{C}_B\).
    Hence, either \(y \notin L_1\), or \(w \notin L_2\).
    If \(y \notin L_1\), then
    \[
        \pr{X_1 \notin L_1} \ge \pr{X_1 = y} = \sqrt{\varepsilon}.
    \]
    If, instead, \(w \notin L_2\), then 
    \[
        \pr{X_2 \notin L_2} \ge \pr{X_2 = w} = \sqrt{\varepsilon}.
    \]
    Hence, 
    \[
    \pr{X_1 \notin L_1} + \pr{X_2 \notin L_2} \ge \sqrt{\varepsilon},
    \]
    proving the claim.
\end{proof}

If we were able to prove a lower bound of $\Omega\left(\min\left\{\log_\Delta n, \log \frac{\varepsilon_0}{\varepsilon} \right\}\right)$ by exploiting only $2$-separation, we would be able to obtain significantly more lower bounds for local problems. Unfortunately, $2$-separation by itself is not sufficient.
\limittwo*
\begin{proof}
    Consider the problem of computing an odd orientation, that is, an orientation of the edges such that all nodes have an odd number of outgoing edges, on $\Delta$-regular graphs, where $\Delta \ge 4$ is even.
    It can be encoded in the node-edge formalism by allowing the edge configuration $\{\mathrm{I},\mathrm{O}\}$, that is, each edge is marked \emph{incoming} by one endpoint and \emph{outgoing} by the other, and by allowing all node configurations of size $\Delta$ containing an odd number of $\mathrm{O}$. Observe that the problem is $2$-separated and $(\Delta/2)$-dominant, but any valid dominance parameter is at least $\Delta/2-1$; hence it is
    not $o(\Delta)$-dominant.

    Observe that, since $\Delta$ is even, each node must have at least one incoming and at least one outgoing edge.
    If both endpoints of an edge label the edge $\mathrm{O}$, or they both label the edge $\mathrm{I}$, this produces an error on the edge. Any $0$-round algorithm must label edges without any coordination with the neighbors. This implies that each edge has an error with probability at least $1/2$. Hence, the expected number of errors $\varepsilon_0$ on the edges must be $\Omega(n)$.

    We give an algorithm that computes, on graphs of girth $\Omega(\log_\Delta n)$, a solution in $O\left(1+ \log_\Delta \frac{\varepsilon_0}{\varepsilon}\right)$ rounds, where $\varepsilon$ is the desired error.

    The algorithm proceeds as follows. Let $T$ be a value to be fixed later.
    Each node samples itself with probability $1 / (\Delta-1)^T$.
    Let $S$ be the set of sampled nodes. We call these nodes \emph{centers}. Each node spends $2T$ rounds to find the nearest node in $S$, if it exists.
    Consider the clustering obtained by putting each node in the cluster of its nearest center (breaking ties by assigning each center an independent uniform random priority from $[0,1]$ and choosing the center with the smallest priority), and by forming a singleton cluster if no center was found (and each singleton acts as the center of its own cluster).
    Orient inter-cluster edges arbitrarily. Then, in each cluster, brute force a solution where there is at most one error on one edge incident to the center (and no other error). Such a solution always exists, because we can process the nodes inwards, and make a node happy by picking the correct orientation for the edge towards the center. After this process, there can be at most one error on the center, which can fix itself by producing one edge error. This takes $O(T)$ rounds, since each cluster has radius at most $2T$.
    
    Note that, for $T$ small-enough compared with the girth of the graph, each node sees a tree. Hence, within $2T$ rounds, each node sees at least $(\Delta-1)^{2T}$ nodes.
    The probability that a node did not find a center within $2T$ rounds is at most:
    \[
        \left(1-\frac{1}{(\Delta-1)^T}\right)^{(\Delta-1)^{2T}} \le e^{-(\Delta-1)^T} \le \frac{1}{(\Delta-1)^T}.
    \]
    Hence, the expected error is at most the expected number of centers, plus the expected number of nodes that did not find a center within $2T$ rounds, which is at most $n/\Delta^{\Omega(T)}$, 
    and by picking $T \in \Theta\left(1+\log_\Delta \frac{\varepsilon_0}{\varepsilon}\right)$ we obtain that the error is at most $\varepsilon$.

\end{proof}

\section{Open Questions}\label{sec:open}
We now discuss possible future directions.
In the following, we use $p$ to denote $\varepsilon / \varepsilon_0$, and to simplify the explanation, we set aside possible differences between $p$ and the standard failure probability.

\paragraph{Alternatives to domination.}
In \Cref{sec:applications}, we proved lower bounds for ``local'' problems $\Pi^L$, by finding suitable $2$-separated $O(1)$-dominant problems $\Pi^G$ that we can solve with small error, given a solution for $\Pi^L$.
The reason why we required $O(1)$-domination is that, when performing the mapping from $\Pi^L$ to $\Pi^G$ we typically get an error in the ballpark of $n / \Delta$ (ignoring logarithmic factors). If we only had a lower bound of $\Omega(\log_\Delta 1/p)$ for $\Pi^G$, we would not get any superconstant lower bound (i.e., $\Omega(\log_\Delta \Delta)$). Hence, for this approach to work, it is crucial to have a lower bound of $\Omega(\log 1/p)$ for $\Pi^G$ (or at least a base of the logarithm that is subpolynomial in $\Delta$).
As shown in \Cref{thm:limit2}, $2$-separation alone is not sufficient to improve the $\Omega(\log_\Delta 1/p)$ lower bound to $\Omega(\log 1/p)$. However, we do not know if $O(1)$-domination is \emph{necessary}.
We would like to understand if there are other types of restriction that we can use to obtain the improvement.
\begin{openquestion}
Which properties can be combined with $2$-separation to obtain lower bounds that are strictly better than $\Omega(\log_\Delta 1/p)$?
\end{openquestion}

\paragraph{A concrete example.}
In order to make the question more concrete, we now describe a simple problem, that we call \emph{exact splitting}, for which we do not know the exact complexity. Consider the problem of coloring the edges red and blue, such that each node has exactly $\Delta/2$ edges of each color, in $\Delta$-regular graphs where $\Delta$ is even.
This problem can be defined in the node-edge formalism as follows. The node constraint allows only $C = \{\mathrm{R},\ldots,\mathrm{R},\mathrm{B},\ldots,\mathrm{B}\}$, where $m_C(\mathrm{R}) = m_C(\mathrm{B}) = \Delta/2$. The edge constraint allows $\{\mathrm{R},\mathrm{R}\}$ and $\{\mathrm{B},\mathrm{B}\}$. This problem is clearly $2$-separated, but it is very far from being $O(1)$-dominant.
However, this looks like a hard problem, for which we do not expect an $O(\log_\Delta 1/p)$ algorithm, but we do not know how to prove an $\omega(\log_\Delta 1/p)$ lower bound with our current techniques.
\begin{openquestion}
    What is the complexity of the exact splitting problem, as a function of $p$?
\end{openquestion}

\paragraph{Alternatives to separation.}
As shown in \Cref{thm:limit1}, $2$-separation is \emph{necessary} for our techniques to work.
Nevertheless, we would like to understand whether we can obtain similar results for problems that are not $2$-separated.
For example, we know that there are problems for which an $\Omega(\log_\Delta n)$ lower bound holds, but that are not $2$-separated. An example of such a problem is the problem of orienting all edges such that all nodes have at least $\Delta-1$ outgoing edges. Such a problem does not admit a solution in high-girth regular graphs, and this is known to imply an $\Omega(\log_\Delta n)$ lower bound for $\Delta \ge 3$, even on trees in which every node has degree either $1$ or $\Delta$, with leaves unconstrained. However, we do not know how to obtain such a lower bound directly via round elimination.
\begin{openquestion}
    How can we obtain $\Omega(\log_\Delta 1/p)$ lower bounds for problems that are not $2$-separated? 
\end{openquestion}

\paragraph{MIS on regular high-girth graphs.}
A natural problem for which we cannot obtain a lower bound via our techniques is Maximal Independent Set on $\Delta$-regular graphs of girth $\Omega(\log_\Delta n)$. We know that obtaining a lower bound of $\Omega\left(\min\left\{\log_\Delta n,\log \Delta\right\}\right)$ is not possible, because there exists a better algorithm \cite{KhourySchild-MIS-26}. However, we would like to understand if a lower bound of $\Omega\left(\min\left\{\log_\Delta n,\log \Delta / \log \log \Delta\right\}\right)$ holds (which is the lower bound known in the non-regular case, via KMW).
We failed to find a problem $\Pi^G$ that we can map to that is $2$-separated and $O(\log \Delta)$-dominant (which would suffice to obtain the lower bound). We would like to understand if such a reduction exists, or if we can find alternatives to separation or to domination that allow us to prove the required lower bound.
\begin{openquestion}
    What is the complexity of MIS in regular high-girth graphs?
\end{openquestion}

\paragraph{Hard problems with large failure probability.}
The approach that we used in \Cref{sec:applications} suggests the following question: what is the complexity of ``hard'' problems, if we allow large failure probability?
As shown, we can use hard problems that we cannot solve in $o(\log 1/p)$ to prove lower bounds for ``local'' problems.
Typically, for problems studied in the LOCAL model, once we see that they are ``global'', we do not investigate them much further. In our work, we show that it can still be very useful to determine their complexity as a function of the required failure probability. Partial results on this topic have been obtained: a recent work studied what is the best local failure probability one can achieve for $2$-coloring cycles, if only $1$ round of communication is allowed \cite{flin-raevskaya-etal-2026-brief-announcement-2-coloring}.
We suggest a systematic study of ``hard'' problems as a function of the required failure probability.
\begin{openquestion}
    Let $\Pi$ be a problem for which we know a deterministic $\Omega(\log_\Delta n)$ lower bound. What is its complexity as a function of the required failure probability?
\end{openquestion}

\paragraph{Different notions of error.}
In order to achieve our results, we needed to introduce a new notion of error, that does not just look at whether a node fails or not, but it counts how many labels are wrong.
There could be other notions of error that allow us to obtain even stronger lower bounds. For example, there could be other notions of error that give $\Omega(\log 1/p)$ lower bounds without requiring $2$-separation or $O(1)$-domination.
In fact, we point out that, even assuming $2$-separation and $O(1)$-domination, by using the standard notion of failure probability, we would have \emph{not} been able to prove $\Omega(\log 1/p)$ lower bounds.
\begin{openquestion}
    Are there other measures of error that we can exploit in order to get automatic $\Omega(\log 1/p)$ lower bounds?
\end{openquestion}

\paragraph{Direct round elimination lower bounds.}
While it is interesting that we can use the hardness of global problems to prove lower bounds for local problems, it would still be interesting to obtain direct lower bounds via round elimination. In more detail, consider the maximal matching problem. We know how to obtain a randomized lower bound of $\Omega\left(\min\left\{\Delta,\log_\Delta \log n\right\}\right)$ via round elimination \cite{balliu-brandt-etal-2021-lower-bounds-for-maximal}, and we know how to obtain a lower bound of  $\Omega\left(\min\left\{\log \Delta,\log_\Delta n\right\}\right)$ if we first map a solution of maximal matching to perfect matching, and then we apply round elimination to perfect matching. However, we still \emph{do not} know how to get the $\Omega\left(\min\left\{\log \Delta,\log_\Delta n\right\}\right)$  lower bound by applying round elimination directly to maximal matching.
For this purpose, we would need to perform a better failure probability analysis for problems that are not $2$-separated  and, more fundamentally, are not even fixed points. Or, maybe, we would need to introduce a new form of round elimination that explicitly tracks the probability that a configuration is used.
\begin{openquestion}
    Can we get an $\Omega\left(\min\left\{\log \Delta,\log_\Delta n\right\}\right)$  randomized lower bound for maximal matching, via a direct application of round elimination?
\end{openquestion}

%% file: figures/perfect-matching.tex
\begin{figure}[tbp]
  \centering
  \begin{subfigure}[t]{.48\textwidth}
    \centering
    \begin{tikzpicture}[x=1.15cm,y=1.15cm]
      \foreach \v/\x/\y in {a/-1.8/1.5,b/1.8/1.5,c/1.8/-1.5,d/-1.8/-1.5,
                            e/-.75/.62,f/.75/.62,g/.75/-.62,h/-.75/-.62}
        \coordinate (\v) at (\x,\y);
      \foreach \u/\v in {a/b,b/c,c/d,d/a,e/f,f/g,g/h,h/e}
        \draw[formalism edge] (\u)--node[formalism label] {$\mathrm U$} (\v);
      \foreach \u/\v in {a/e,b/f,c/g,d/h}
        \draw[formalism matched] (\u)--node[formalism label,text=figureblue] {$\mathrm M$} (\v);
      \foreach \v in {a,c,f,h} \node[formalism white] at (\v) {};
      \foreach \v in {b,d,e,g} \node[formalism black] at (\v) {};
    \end{tikzpicture}
    \[
      \begin{aligned}
        \constr{C}_W&=\bigl\{\{\mathrm M,\mathrm U,\mathrm U\}\bigr\},\\
        \constr{C}_B&=\bigl\{\{\mathrm M,\mathrm U,\mathrm U\}\bigr\}.
      \end{aligned}
    \]
    \caption{Black-white formalism.}
  \end{subfigure}\hfill
  \begin{subfigure}[t]{.48\textwidth}
    \centering
    \begin{tikzpicture}[x=1.15cm,y=1.15cm]
      \foreach \v/\x/\y in {a/-1.8/1.5,b/1.8/1.5,c/1.8/-1.5,d/-1.8/-1.5,
                            e/-.75/.62,f/.75/.62,g/.75/-.62,h/-.75/-.62}
        \coordinate (\v) at (\x,\y);
      \foreach \u/\v in {a/b,b/c,c/d,d/a,e/f,f/g,g/h,h/e}
        \draw[formalism edge] (\u)--
          node[formalism label,pos=.23] {$\mathrm U$}
          node[formalism label,pos=.77] {$\mathrm U$} (\v);
      \foreach \u/\v in {a/e,b/f,c/g,d/h}
        \draw[formalism matched] (\u)--
          node[formalism label,pos=.23,text=figureblue] {$\mathrm M$}
          node[formalism label,pos=.77,text=figureblue] {$\mathrm M$} (\v);
      \foreach \v in {a,b,c,d,e,f,g,h} \node[formalism white] at (\v) {};
    \end{tikzpicture}
    \[
      \begin{aligned}
        \constr{C}_W&=\bigl\{\{\mathrm M,\mathrm U,\mathrm U\}\bigr\},\\
        \constr{C}_B&=\bigl\{\{\mathrm M,\mathrm M\},\{\mathrm U,\mathrm U\}\bigr\}.
      \end{aligned}
    \]
    \caption{Node-edge formalism.}
  \end{subfigure}
  \caption{Perfect matching in the two formalisms, illustrated on the same $3$-regular graph.
    Thick blue edges are matched; all other edges are unmatched.
    In (a), each edge has one label, and every white and black node sees exactly one $\mathrm M$.
    In (b), each node-edge pair has a label; every node sees exactly one $\mathrm M$,
    and the two labels on each edge agree.
    Subdividing each edge in (b) by a black node gives its black-white representation,
    with white degree $3$ and black degree $2$.}
  \label{fig:perfect-matching-formalisms}
\end{figure}

%% file: figures/perfect-h-packing.tex
\begin{figure}[tbp]
  \centering
  \begin{subfigure}{\textwidth}
    \centering
    \begin{tikzpicture}[x=1cm,y=1cm]
      % The marked root is part of the rooted, color-preserving type.
      \foreach \x/\lab in {0/1,3/2,6/3,9/4}
        \node[font=\small] at (\x,.55) {$\sigma(t)=\lab$};
      \node[formalism white] at (0,0) {};
      \draw[formalism root] (0,0) circle (5pt);
      \draw[formalism edge] (3,0)--(3,-.5);
      \node[formalism black] at (3,0) {};
      \node[formalism white] at (3,-.5) {};
      \draw[formalism root] (3,0) circle (5pt);
      \foreach \x in {-.5,.5} {
        \draw[formalism edge] (6,0)--(6+\x,-.5)--(6+\x,-1);
        \node[formalism black] at (6+\x,-.5) {};
        \node[formalism white] at (6+\x,-1) {};
      }
      \node[formalism white] at (6,0) {};
      \draw[formalism root] (6,0) circle (5pt);
      \draw[formalism edge] (9,0)--(9,-.5);
      \foreach \x in {-.5,.5} {
        \draw[formalism edge] (9,-.5)--(9+\x,-1)--(9+\x,-1.5);
        \node[formalism black] at (9+\x,-1) {};
        \node[formalism white] at (9+\x,-1.5) {};
      }
      \node[formalism black] at (9,0) {};
      \node[formalism white] at (9,-.5) {};
      \draw[formalism root] (9,0) circle (5pt);
    \end{tikzpicture}
    \caption{The four rooted subtree types; a blue ring marks each root.}
  \end{subfigure}

  \medskip
  \begin{subfigure}[t]{.48\textwidth}
    \centering
    \begin{tikzpicture}[x=1.05cm,y=1.05cm]
      \foreach \v/\x/\y in {c/0/0,a/-1.5/-1.6,b/1.5/-1.6,s/0/-1.9,
                            d/-2.25/-3.2,e/-.75/-3.2,f/.75/-3.2,g/2.25/-3.2}
        \coordinate (\v) at (\x,\y);
      \foreach \u/\v/\p/\q in {c/a/4/3,c/b/4/3,c/s/2/1,
                                a/d/2/1,a/e/2/1,b/f/2/1,b/g/2/1} {
        \coordinate (m\v) at ($(\u)!.5!(\v)$);
        \draw[formalism edge] (\u)--node[formalism label] {$\p$} (m\v)
          --node[formalism label] {$\q$} (\v);
        \node[formalism black] at (m\v) {};
      }
      \foreach \v in {c,a,b,s,d,e,f,g} \node[formalism white] at (\v) {};
      \node[above,font=\small] at (c) {$c$};
    \end{tikzpicture}
    \caption{The incidence graph $H'$, rooted at its center $c$.}
  \end{subfigure}\hfill
  \begin{subfigure}[t]{.48\textwidth}
    \centering
    \begin{tikzpicture}[x=1.05cm,y=1.05cm]
      \foreach \v/\x/\y in {c/0/0,a/-1.5/-1.6,b/1.5/-1.6,s/0/-1.9,
                            d/-2.25/-3.2,e/-.75/-3.2,f/.75/-3.2,g/2.25/-3.2}
        \coordinate (\v) at (\x,\y);
      \foreach \u/\v/\p/\q in {c/a/4/3,c/b/4/3,c/s/2/1,
                                a/d/2/1,a/e/2/1,b/f/2/1,b/g/2/1}
        \draw[formalism edge] (\u)--
          node[formalism label,pos=.23] {$\p$}
          node[formalism label,pos=.77] {$\q$} (\v);
      \foreach \v in {c,a,b,s,d,e,f,g} \node[formalism white] at (\v) {};
      \node[above,font=\small] at (c) {$c$};
    \end{tikzpicture}
    \caption{The induced labels on the node-edge pairs of $H$.}
  \end{subfigure}

  \smallskip
  \begin{equation*}
    \begin{aligned}
      \Sigma&=\{1,2,3,4,\bot\},\\
      \constr{C}_W&=\bigl\{\{4,4,2,\bot\},\{3,2,2,\bot\},
                                      \{1,\bot,\bot,\bot\}\bigr\},\\
      \constr{C}_B&=\bigl\{\{4,3\},\{2,1\},\{\bot,\bot\}\bigr\}.
    \end{aligned}
  \end{equation*}
  \caption{Encoding perfect $H$-packing for an eight-node tree $H$ that contains the same subtree multiple times.
    In (b), each edge is labeled by the type of the subtree below it.
    Types $3$ and $4$ each occur twice, while types $1$ and $2$ each occur five times,
    including along the short middle branch.
    The constraints below the diagrams are for perfect $H$-packing in a $4$-regular input graph.
    Edges outside the copies of $H$ receive $\bot$ at both endpoints.}
  \label{fig:perfect-h-packing-encoding}
\end{figure}